%% file: example_paper.tex
\documentclass{article}

\usepackage{microtype}
\usepackage{graphicx}
\usepackage{subcaption}
\usepackage{wrapfig}
\usepackage{subfiles}
\usepackage{booktabs}% for professional tables
\input{math_commands.tex}

\usepackage{hyperref}

\usepackage[numbers]{natbib}
\usepackage[margin=1in]{geometry}
\usepackage{amsmath}
\usepackage{amssymb}
\usepackage{mathtools}
\usepackage{amsthm}

\usepackage[capitalize,noabbrev]{cleveref}
\usepackage{amsbsy,amsgen,amscd,amssymb,stmaryrd,mdframed}
\theoremstyle{plain}
\newtheorem{theorem}{Theorem}[section]

\newtheorem{lemma}{Lemma}
\newtheorem{corollary}{Corollary} %[theorem]{Corollary}
\theoremstyle{definition}
\newtheorem{definition}[theorem]{Definition}

\newtheorem{remark}[theorem]{Remark}
\newtheorem*{openquestion}{Open Question}

\usepackage{framed,color}
\definecolor{shadecolor}{rgb}{.9, .9, .9}
\makeatletter

\colorlet{shadecolor}{blue!10}
\newcommand{\thistheoremname}{}
\newmdtheoremenv[
outerlinewidth=2,
roundcorner=10 pt,
leftmargin=0,
rightmargin=0,
backgroundcolor=blue!5,
outerlinecolor=blue!70!black,
innerleftmargin=.4em,
    innerrightmargin=.4em,
    innertopmargin=.4em,
    innerbottommargin=.4em,
    skipabove=.3em,
        skipbelow=.3em,
hidealllines=true,
splittopskip=\topskip,
ntheorem=true]{genericthm}{\thistheoremname}

\DeclareFontFamily{OT1}{pzc}{}
\DeclareFontShape{OT1}{pzc}{m}{it}{<-> s * [1.200] pzcmi7t}{}
\DeclareMathAlphabet{\mathpzc}{OT1}{pzc}{m}{it}
\usepackage{xcolor,colortbl}

\makeatletter
\newtheorem*{rep@theorem}{\rep@title}
\newcommand{\newreptheorem}[2]{%
	\newenvironment{rep#1}[1]{%
		\def\rep@title{#2 \ref{##1}}%
		\begin{rep@theorem}}%
		{\end{rep@theorem}}}
\makeatother

\newreptheorem{theorem}{Theorem}
\newreptheorem{lemma}{Lemma}
\newreptheorem{corollary}{Corollary}
\newtheorem*{question}{Question}
\newtheorem{defn}{Definition}

\usepackage{lscape}
\usepackage{makecell}
\usepackage{thmtools,thm-restate}

\renewcommand*{\Pr}{\mathrm{Pr}}

\usepackage{algorithm,algorithmic}
\newcounter{protocol}
\makeatletter
\newenvironment{protocol}[1][htb]{%
  \let\c@algorithm\c@protocol
  \renewcommand{\ALG@name}{Protocol}% Update algorithm name
  \begin{algorithm}[#1]%
  }{\end{algorithm}
}
\makeatother
\usepackage[textsize=tiny]{todonotes}

\allowdisplaybreaks
\title{Semi Bandit Dynamics in Congestion Games: Convergence to Nash Equilibrium and No-Regret Guarantees.}

\author{
\and\and
\textbf{Ioannis Panageas}$^\star$ \\
Computer Science\\ UC Irvine
\and
\textbf{Stratis Skoulakis}$^\star$ \\
LIONS \\
EPFL 
\and
\textbf{Luca Viano}$^\star$ \\
LIONS \\
EPFL\and
\and\and
\textbf{Xiao Wang} \\
ITCS\\ SUFE
\and
\textbf{Volkan Cevher} \\
LIONS \\
EPFL
}
\begin{document}
%\printAffiliationsAndNotice{}  % leave blank if no need to 
\maketitle
\begin{abstract}

In this work, we introduce a new variant of online gradient descent, which provably converges to Nash Equilibria and simultaneously attains sublinear regret for the class of congestion games in the semi-bandit feedback setting. Our proposed method admits convergence rates depending only polynomially on the number of players and the number of facilities, but not on the size of the action set, which can be \let\thefootnote\relax\footnotetext{$^\star$ Equal Contribution.} exponentially large in terms of the number of facilities. 
Moreover, the running time of our method has polynomial-time dependence on the implicit description of the game. 
%Our analysis exploits techniques from convex geometry, in particular  Caratheodory's theorem and recent advances in non-convex stochastic optimization.
As a result, our work answers an open question from \cite{du22}. 
\end{abstract}

\section{Introduction}

Congestion games is a class of multi-agent games at which $n$ selfish agents compete over a set of $m$ resources. Each agent selects a subset of the resources and her cost depends on the load of each of the selected resources (number of other agents using the resource). For example in \textit{Network Congestion Games}, given a graph, each agent $i$ wants to travel from a starting vertex $s_i$ to a target position $t_i$ and thus needs to select set of edges forming an $(s_i,t_i)$-path. Due to their numerous applications %in routing, electrical grids, rate allocation e.t.c.
congestion games have been extensively studied over the years \cite{KoutsoupiasP99WorstCE,roughgarden2002bad,christodoulou,Fotakis2005226,Schafer10,Roughgarden09}.

It is well-known that congestion games always admit a Nash Equilibrium (NE) which is a \textit{steady state} at which no agent can unilaterally deviate without increasing her cost. At the same time, a long line of research studies the convergence properties to NE of \textit{game dynamics} (e.g. no-regret, best response, fictitious play) at which the agents of a congestion game iteratively update their strategies based on the strategies of the other agents
in their attempt to minimize their individual cost.

In most real-world scenarios, agents do not have access to the strategies of the other agents (\textit{full-information feedback}) and are only informed on the loads/cost of the resources they selected at each round. For example a driver learns only the congestion on the highways that she selected and not the congestion of the alternatives that she did not select. This type of feedback is called \textit{semi-bandit feedback} and has been extensively studied in the context of online learning. Motivated by the above, \cite{du22} in their recent work investigate the following question:

\begin{question}\cite{du22}
\textit{Are there update rules under (semi)-bandit feedback, that once adopted by all agents of a congestion game, the overall system converges to a Nash Equilibrium with rate that is \textbf{independent} of the number of possible strategies?} 
\end{question}
We note that in congestion games the number of possible strategies can be exponentially large with respect to the game description. For example the number of $(s,t)$-paths
can be of the order $\mathcal{O}\left(2^{\Theta(m
)}\right)$ with respect to the number of edges $m$. \cite{du22} provide an update rule (based on the Frank-Wolfe method)
that once adopted by all $n$ agents, the overall system requires $\mathcal{O}(n^{12} m^9/\epsilon^6)$ time-steps (samples) to reach an $\epsilon$-approximate NE.

Despite its fast convergence properties, the update rule of \cite{du22} is not aligned with the selfish nature of the agents participating in a congestion game. This is because their method does not provide any kind of guarantees on the \textit{regret} of the agents adopting it. As a result, \cite{du22} posed the following open question.
\begin{openquestion}\cite{du22}
\textit{Are there update rules under (semi)-bandit feedback that}
\begin{enumerate}
    \item \textit{provide (adversarial) no-regret guarantees to any agent that adopts them} and
    \item \textit{once adopted by all agents, the overall system converges to Nash Equilibrium with rate independent on the number of strategies?}
\end{enumerate}
\end{openquestion}
The term \textit{no-regret} refers to the fact that
the time-average cost of any agent, adopting the update rule, is upper bounded by the time-average cost
of the \textit{best fixed path in hindsight} (no matter how the other agents select their strategies). Due to their remarkable guarantees, no-regret algorithms have been the standard choice for modeling selfish behavior in non-cooperative environments \cite{EMN09}.

\textbf{Our Contribution.}
In this work, we provide a positive answer to the above open question, while we improve upon the results of \cite{du22} in two important aspects. Specifically, we propose a semi-bandit feedback online learning algorithm, called Semi-Bandit Gradient Descent with Caratheodory Exploration ($\mathrm{SBGD-CE}$), with the following properties: 
\begin{itemize}
\item \textit{\textbf{No-regret guarantees}}: Any agent adopting $\mathrm{SBGD}$-$\mathrm{CE}$ admits at most $\mathcal{\tilde{O}}(m^{2}T^{4/5})$ regret no matter how the other agents select their strategies.
\item \textit{\textbf{Convergence to NE}}: If $\mathrm{SBGD}-\mathrm{CE}$ is adopted by all agents, the overall system reaches an $\epsilon$-approximate NE within $\mathcal{O}(n^{6.5}m^{7}/\epsilon^5)$ time steps improving the $\mathcal{O}(n^{12} m^9/\epsilon^6)$ bound of \cite{du22}. 

\item \textit{\textbf{Polynomial Update Rule}}: 
The update rule of $\mathrm{SBGD}$-$\mathrm{CE}$ runs in polynomial-time with respect to the \textit{implicit description} of the strategy space (see Section~\ref{s:implicit}). On the other hand, the update rule of \cite{du22} requires linear time and space complexity w.r.t the number of strategies. Thus, in many interesting settings such as \textit{network congestion games}, the time and space complexity of Frank-Wolfe with Exploration II \cite{du22} is exponential w.r.t the number of edges. %while $\mathrm{SBGD}$-$\mathrm{CE}$ runs in polynomial time.
\end{itemize}
\begin{remark}[Notion of convergence]\label{rem:last} As in \cite{du22,leonardos2022global,DWZJ22,Ana22},the notion of convergence that we use the is so-called \textit{best-iterate convergence}. Mathematically speaking, the time-averaged exploitability (defined as the sum of best deviation minus chosen strategy per agent) is bounded by $\epsilon$ (see Theorem \ref{thm:convergence_to_nash}). From a \textit{game-theoretic} point of view, bset-iterate convergence implies that with high probability \textit{almost all iterates} are $\mathcal{O}(\epsilon)$-Mixed NE. From a \textit{learning} point of view, best-iterate convergence implies that we can learn an approximate NE of an \textit{unknown} congestion game by considering the strategy profile at a iterate $t \sim \mathrm{Unif}(1,\ldots,T)$ (see Corollary~\ref{c:markov}). The term "\textit{best-iterate convergence}" might not be the most descriptive for the above, however it is the one most commonly used in the literature. 
\end{remark}

\begin{table}[t]
\centering
\caption{ Comparison with previous related works. $^\star$See \Cref{remarkT34} for regret bound $\mathcal{O}(m^{3/2}T^{3/4})$ under different step size and exploration parameter choices. $\dagger$By convergence to NE we mean best-iterate convergence as explained in Remark \ref{rem:last}. We note that all the known results that provide rates use the same notion of convergence (as presented in the table).
\label{table:results}}
\resizebox{\textwidth}{!}{\begin{tabular}{@{}cccc@{}}
\toprule
\textbf{Method} & \textbf{Adversarial Regret} & \textbf{Convergence$\dagger$ to NE} &  \textbf{Running Time} \\ \midrule
%OMSD \cite{audibert2014regret} & $\mathcal{O}\left( m \sqrt{T}\right)$  & Not Established  &  Poly. in Implicit Description        \\ \midrule
%FPL+GR \cite{neu2016importance} & $\mathcal{O}\left( m^{3/2} \sqrt{T}\right)$  & Not Established  &  Poly. in Implicit Description        \\ \midrule
IPPG \cite{leonardos2022global} &Not Established  & $\mathcal{O}\left(n2^{\Theta(m)}m/\epsilon^6\right)$  &  Exp. in Available Resources        \\ \midrule
IPGA \cite{DWZJ22} & Not Established  &  $\mathcal{O}\left(n^32^{\Theta(m)}m^5 /\epsilon^5\right)$ & Exp. in Available Resources \\ \midrule
FW with Exploration $\mathrm{II}$ \cite{du22} &   Not Established   & $\mathcal{O}\left(n^{12}m^{9}/\epsilon^6\right)$  & Exp. in Available Resources \\ \midrule
$\mathrm{SBGD}-\mathrm{CE}$ (this work) & $\mathcal{O}(m^2T^{4/5})^\star$   & $\mathcal{O}\left(n^{6.5}m^{7}/\epsilon^5\right)$  & Poly. in Implicit Description \\ \bottomrule
\end{tabular}}
%\caption{\textcolor{blue}{Luca's comment: What about adding also FPL+GR and OMSD in the table ? It would higligth that our algorithm is the only one with guarantees in both settings.}}
\end{table}

\textbf{Our Techniques and Related Work}
The fundamental difficulty in designing no-regret online learning algorithms under (semi)-bandit feedback is to guarantee that each strategy is sufficiently explored. Unfortunately, standard bandit algorithms such as $\mathrm{EXP}3$ \cite{ACFS02} result in exponentially large in $m$ regret bounds, e.g. $\mathcal{O}\left(2^{\Omega(m)}\sqrt{T}\right)$, as well as time and space complexity.       
A long line of research in the context of \textit{combinatorial bandits} provides no-regret algorithms with polynomial dependence w.r.t to $m$ on the regret, while many of those algorithms can be efficiently implemented \cite{AK04, VHK07,GLLO07,BCK12,CL12,vempala, NB13, ASL14}. For example, \cite{BCK12} provide an online learning algorithm with regret $\mathcal{O}(m\sqrt{T})$. However, in order to overcome the exploration problem, these algorithms use involved techniques (e.g. barycentric spanners or entropic projections \cite{AK04,BCK12}) which introduce major technical difficulties in their \textit{multi-agent analysis}. To the best of our knowledge, none of these algorithms guarantees convergence to NE in congestion games once adopted by all agents.
\begin{remark}
We highly remark that the \textit{no-regret property} does not imply convergence to Nash Equilibrium in congestion/potential games \cite{CHM17,BR20}. No-regret dynamics are guaranteed to converge in Coarse Correlated Equilibrium that is a strict superset of NE and that can even contain strictly dominated strategies \cite{VZ13}.  
\end{remark}

On the somehow opposite front, recent works studying the convergence properties of semi-bandit game dynamics in potential games use \textit{explicit} exploration schemes at which each strategy is selected with a small probability \cite{leonardos2022global,DWZJ22}. However such exploration schemes lead to convergence rates that scale polynomially on the number of strategies (can be exponential w.r.t to $m$ in congestion games). \cite{du22} combine an explicit exploration scheme with the Frank-Wolfe method and establish that the resulting convergence rate (number of samples) to NE depends only polynomially in $m$. As mentioned above, the update rule of \cite{du22} does not guarantee the no-regret property in the adversarial case while its update rule is of exponential time and space. Table~\ref{table:results} concisely present the above mentioned results.

In order to solve the exploration problem with schemes that are simple enough to analyze in the multi-agent case, we introduce the notion of \textit{Bounded-Away Description Polytope}. These polytopes are subsets of description polytopes, the extreme points of which correspond to the available strategies and additionally impose lower bounds on the fractional selection of each resource. Our $\mathrm{SBGD-CE}$ method is based into running Online Gradient Descent \cite{Z07} while projecting into a \textit{time-expanding} Bounded-Away Description Polytope. At each round, $\mathrm{SBGD-CE}$ also uses the Caratheodory Decomposition to (randomly) select a valid set of resources. By extending the analysis of Online Gradient Descent as well as of Stochastic Gradient Descent constrained to \textit{time-varying feasibility sets}, we establish no-regret guarantees as well as fast converge to NE.

%\textbf{Futher Related work on learning in congestion games.} In the full feedback setting, well-studied class of no-regret algorithms include follow-the-regularized-leader (FTRL) and more specifically multiplicative weights update and gradient descent. These algorithms have been proven quite successful in converging to Nash equilibria in congestion games \cite{PPP17}, \cite{HeliouCM17}, \cite{Ana22}, . 

\textbf{Further Related Work} \cite{Ana22} establish best-iterate convergence rates to NE in congestion games, though \textit{full-information feedback} is assumed and the rates depends on the strategy space of each agent. 
\cite{CMS10} and \cite{PPP17,HCM17} prove \textit{asymptotic last-iterate convergence} of no-regret dynamics under \textit{full-information feedback} and \textit{bandit-feedback} respectively in potential/congestion games. To the best of our knowledge there do not exist \textit{last-iterate convergence rates} for congestion games (even with exponential dependence on the number of resources) unless the initial condition is close enough to an equilibrium. Even the well-studied \textit{full-information better-response dynamics} is only known to converge in the \textit{best-iterate sense} \cite{CS07}. \cite{du22} also provide provide convergence guarantees for congestion games under \textit{bandit feedback} with slightly worse rates that the one presented in Table~\ref{table:results}. \cite{VAM21} study accelerated methods to converging to \textit{Wardrop Equilibrium} in network congestion games. Other works studying no-regret dynamics beyond congestion games include \cite{PS14,MS17,CHM17,MPP18,BLM18,MZ19, VG20,GVM21}.

\section{Preliminaries and Our Results}
\label{sec:preliminaries}

\subsection{Congestion Games}
A \textit{congestion game} is composed by $n$ selfish agents and 
a set of resources $E$ with $|E| = m$. The strategy of each agent $i\in [n]$ is a subset of resources $p_i \in \mathcal{P}_i$ where $\mathcal{P}_i \subseteq 2^E$ is the strategy space of agent $i$. The set $\mathcal{P}$ denotes all joint strategy profiles, $\mathcal{P} :=  \mathcal{P}_1 \times \ldots \times\mathcal{P}_n$. Given a strategy profile $p \in \mathcal{P},$ we use the notation $p :=(p_i,p_{-i})$ where $p_i$ captures the strategy of agent $i$ and $p_{-i}$ denotes the strategies of all agents but $i$.

The \textit{load} of a resource $e\in E$ under the strategy profile $p:= (p_1,\ldots,p_n) \in \mathcal{P}$ is denoted by $\ell_e(p)$ and equals the number of agents using resource $e \in E$, i.e.,
\begin{equation}\label{eq:load}
\ell_e(p_1,\ldots, p_{n}) \triangleq \sum^n_{i=1} \mathds{1}\left[e\in p_i\right].
\end{equation}

Each resource $e$ admits a positive and non-decreasing cost function $c_e:~\mathbb{N} \mapsto \mathbb{R}_{\geq 0}$ where $c_e(\ell)$ denotes the congestion cost of $e \in E$
under load $\ell \in \mathbb{N}$. Additionally, we set $c_{\max} := \max_{e\in E} c_e(n)$.

Given a strategy profile $p := (p_1,\ldots,p_n) \in \mathcal{P}$, the cost of agent $i$ is defined as
        \[C_i(p) \triangleq \sum_{e \in p_i} c_e\left(\ell_e(p)\right).\]
        
\begin{defn}[\textbf{Nash Equilibrium}]
 A path selection $p=(p_1,\ldots,p_n) \in \mathcal{P}$ is an $\epsilon$-approximate \textit{Pure Nash Equilibrium} if and only if for each agent $i \in [n]$,
\begin{equation*}
    C_i(p_i, p_{-i}) \leq C_i(p'_i, p_{-i}) + \epsilon \quad \text{for all } p'_i \in \mathcal{P}_i.
\end{equation*}
A prob. distribution $\pi^\star:= (\pi^\star_1,\ldots,\pi^\star_n) \in \Delta(\mathcal{P}_1)\times \dots \times \Delta(\mathcal{P}_n)$ is an $\epsilon$-approximate \textit{Mixed Nash Equilibrium} (MNE) if and only if for each agent $i \in [n]$,
\begin{equation*}
    \E_{\pi^\star_i,\pi^\star_{-i}}\left[ C_i(p_i, p_{-i})\right] \leq \E_{\pi'_i,\pi^\star_{-i}}\left[ C_i(p_i, p_{-i})\right] + \epsilon .
\end{equation*}
for all $\pi_i' \in \Delta(\mathcal{P}_i)$.
For convenience of notation, we will use the shorthand $c_i(\pi_i, \pi_{-i}) := \E_{\pi_i,\pi_{-i}}\left[ C_i(p_i, p_{-i})\right]$.
\end{defn}

\begin{theorem}[Folkore]
Congestion games always admit a Pure NE equilibrium $p^\star \in \mathcal{P}$.
\end{theorem}

\subsection{Implicit Description of the strategy Space}\label{s:implicit}
The strategy space $\mathcal{P}_i$ can be exponentially large w.r.t the number of resources $E$ ($|\mathcal{P}_i| \leq 2^{\Theta(m)}$). For example in \textit{network congestion games}, $\mathcal{P}_i$ is the set of possible $(s_i,t_i)$-paths in a given directed graph $G(V,E)$ where $s_i \in V$ is the starting node of agent $i$ and $t_i \in V$ is her destination. Typically the number of possible $(s_i,t_i)$ paths is exponential in the number of edges. 

Exponentially large strategy spaces can be described through the following \textit{implicit polytopal description} \cite{KS21}. Given a strategy $p_i \in \mathcal{P}_i$, consider
its equivalent description as $\{0,1\}^m$ vector 
\[x_
{p_i} \triangleq \{x \in \{0,1\}^m~:~x_e = 1 ~\text{iff}~  e\in p_i\}.\] 
Consider the set $\hat{\mathcal{P}}_i := \{x_{p_i} \in \{0,1\}^m:~\text{for }p_i \in \mathcal{P}_i \}$ and the polytope $\mathcal{X}_i := \mathrm{conv}(\mathcal{\hat{P}}_i)$ denoting the convex hull of $\mathcal{\hat{P}}_i$. The polytope $\mathcal{X}_i$ admits the alternative description $\mathcal{X}_i := \{x \in [0,1]^m:~A_i \cdot x \leq b_i\}$ where $A_i \in \mathbb{R}^{r \times m}$ and $b_i \in \mathbb{R}^r$. Thus $\mathcal{P}_i$ can be \textit{implictly} described as the extreme points of a set $\mathcal{X}_i := \{x \in [0,1]^m:~A_i \cdot x \leq b_i\}$ that can be described with $\mathcal{O}(r m)$ fractional numbers.

In many classes of congestion games, the implicit polytopal description of $\mathcal{P}_i$ is polynomial in the number of resources ($r := \mathrm{poly}(m)$) while $|\mathcal{P}_i| = 2^{\Theta(m)}$ \cite{KS21}. For example in \textit{directed acyclic graphs} ($\mathrm{DAG}$s) the number all possible $(s_i,t_i)$-paths can be $2^{\Theta(m)}$ while the set of $(s_i,t_i)$-paths can be equivalently described as the extreme points of the following \textit{path polytope} (see Appendix~\ref{app:DAG}),
\begin{align*}
    \mathcal{X}_i \triangleq \bigg\{ \fcost&\in[0,1]^m : \sum_{e \in \mathrm{Out}(s_i)} x_e = 1, \sum_{e \in \mathrm{In}(t_i)} x_e = 1,\\ & \sum_{e \in\mathrm{In}(v)} x_e =  \sum_{e \in\mathrm{Out}(v)} x_e \quad \forall v \in V \setminus \bc{s_i, t_i}  \bigg\}
\end{align*}
where $\mathrm{In}(v),\mathrm{Out}(v) \subseteq E$ denote the incoming, outgoing edges respectively of the node $v \in V$. 
\begin{remark}\label{r:1}
One can always compute an implicit polytopal description $\mathcal{X}_i$ given an explicit description of $\mathcal{P}_i$. In the remaining paper, we assume access to the \textit{implicit polytopal description} $\mathcal{X}_i := \{x \in [0,1]^m:~A_i \cdot x \leq b_i\}$ where $A_i \in \mathbb{R}^{r_i \times m}$ and $b_i \in \mathbb{R}^{r_i}$. We note that the regret bounds and the convergence rates to NE of our proposed algorithm (Algorithm~\ref{alg:alg1}) only depend on the $n$ and $m$ and are totally independent of $\max_{i \in [n]}r_i$. The exact same holds for the convergence guarantees of the update rule proposed by \cite{du22}. On the other hand, the running time of Algorithm~\ref{alg:alg1} is polynomial in $r_i$ and $m$ while the time and space complexity of algorithm of \cite{du22} scales linearly with $\max_{i \in [n]}|\mathcal{P}_i|$ even if $r_i= \mathrm{poly}(m)$.%For example in the case of \textit{network congestion games} in DAGs the update rule of \cite{du22} admits $2^{\Theta(m)}$ time complexity while Algorithm~\ref{alg:alg1} is polynomial in $m$ (see Section~\ref{sec:experiments}). 
\end{remark}

\subsection{Semi-Bandit Learning Dynamics}
%Nash Equilibrium describes a steady-state of a system
%at which no agent wants to deviate from its selected strategy. However it is not clear whether the limiting behavior of natural game dynamics lead to NE.
In game dynamics with \textit{semi-bandit feedback}, each agent iteratively updates her strategies based on the congestion cost of the previously selected resources so as to minimize her overall experienced cost. Semi-bandit dynamics in congestion are described in Protocol~\ref{def:games_protocol}.
\begin{protocol}[h]
			\caption{Semi-Bandit Game Dynamics}
			\label{def:games_protocol}
			\begin{algorithmic}[1]
				\FOR{each round $t = 1, \ldots, T$}

				\STATE Each agent $i \in [n]$ (randomly) selects a strategy $p^t_i \in \mathcal{P}_i$ and suffers cost \[C_i(p^t_i, p^t_{-i}):= \sum_{e \in p_i^{t}} c_e^t(p_i^t,p_{-i}^t)\]
    \STATE Each agent $i \in [n]$ \textbf{learns only} the congestion costs $c_e(p_i^t,p^t_{-i})$ of her selected resources $e \in p_i^t$.
    \smallskip
%\STATE Each agent $i \in [n]$ uses the above semi-bandit feedback to update its decision in the decision.
				\ENDFOR
   \end{algorithmic}
   
		\end{protocol}
  
It is not clear how a selfish agent $i$ should update her strategy to minimize her overall congestion cost since the loads depend on strategies of the other agents that can arbitrarily change over time. As a result, agent $i$ tries to minimize her experienced cost under the worst-case assumption that the cost of the resources $c_e^t$ are selected by a \textit{malicious adversary}. We refer to the latter online learning setting as \textit{Online Resource Selection},  described in Protocol~\ref{def:path_selection_protocol}.

\begin{protocol}[h]
			\caption{Online Resource Selection}
			\label{def:path_selection_protocol}
			\begin{algorithmic}[1]
				\FOR{each round $t = 1, \ldots ,T$}
                    \smallskip
				\STATE Agent $i$ selects a prob. distribution $\pi_i^t \in \Delta(\mathcal{P}_i)$.
    \smallskip
    \STATE An adversary selects a cost function $c^t: E \mapsto \mathbb{R}_{\geq 0}$.
    \smallskip
 \STATE Agent $i$ samples a path $p^t_i \sim \pi_i^t$ and suffers cost
        \[C_i(p^t_i, c^t) := \sum_{e \in p_i^t} c_e^t.\]
        \item Agent $i$ learns the costs $c_e^t$ for all resources $e \in p_i^t$ and updates $\pi_i^{t+1} \in \Delta(\mathcal{P}_i)$.
				\ENDFOR
   \end{algorithmic}
   
		\end{protocol}
  
A \textit{semi-bandit online learning algorithm} $\mathcal{A}$ for the Online Resource Selection selects a strategy $p_i^t \in \mathcal{P}_i$ based on the observed costs of selected resources in rounds before $t$. The quality of an algorithm $\mathcal{A}$ is measured through the notion of \textit{regret} capturing the overall cost of algorithm $\mathcal{A}$ with respect to the overall cost of the \textit{best strategy in hindsight}. 

\begin{definition}\label{d:no-regret}
The regret of an online learning algorithm $\mathcal{A}$ is defined as $\mathcal{R}_{\mathcal{A}}(T):=$
\[ \max_{c^1,\ldots,c^T}\left[\sum^T_{t=1} \E_{\pi^t_i}\left[C_i(p^t_i,c^t)\right] -  \min_{p^\star_i\in\mathcal{P}} \sum^T_{t=1} C_i(p^\star_i,c^t)\right]\]
In case $\mathcal{R}_{\mathcal{A}}(T) = o(T)$, i.e., it is sublinear in $T$, the algorithm $\mathcal{A}$ is called \textit{no-regret}.
\end{definition}

In the context of congestion games, if agent $i$ adopts a no-regret algorithm $\mathcal{A}$, her time-averaged cost approaches the cost of the optimal path with rate $\mathcal{R}_{\mathcal{A}}(T)/T \rightarrow 0$ no matter what the strategies of the other agents are.  

We conclude the section with the main result of our work.

\textbf{Main Result} \textit{There exists a no-regret semi-bandit online learning algorithm that admits $\mathcal{R}_{\mathcal{A}}(T) = \mathcal{O}(m^2 T^{4/5})$ regret for Online Resource Selection (Algorithm~\ref{alg:alg1}). Moreover if Algorithm~\ref{alg:alg1}
is adopted by all agents, then the agents converge to an $\epsilon$-Mixed NE after $\mathcal{O}(n^{6.5}m^{7}/\epsilon^5)$ time-steps.}

In Section~\ref{s:algo_presentation} we present our proposed semi-bandit online learning algorithm, called \textit{Semi-Bandit Gradient Descent with Caratheodory Exploration} (Algorithm~\ref{alg:alg1}). In Section~\ref{s:algo_presentation} we also present its regret guarantees (Theorem~\ref{t:regret}) and its convergence properties (Theorem~\ref{thm:convergence_to_nash} and Corollary~\ref{c:markov}). In Section~\ref{sec:first_thm_sketch} we provide the main steps for establishing the no-regret properties of Algorithm~\ref{alg:alg1} while in Section~\ref{sec:second_thm_sketch} we present the main ideas of the convergence proof. Finally in Section~\ref{sec:experiments} we experimentally evaluate our algorithm in network congestion games.

\section{Semi-Bandit Gradient Descent with Caratheodory Exploration}\label{s:algo_presentation}
In this section, we present our algorithm called \textit{Bandit Gradient Descent with Caratheodory Exploration} for the Online Path Selection Problem.

\subsection{Exploring via Caratheodory Decomposition}
{To avoid the exponentially large strategy space, we re-parametrize the problem using a fractional selection $x^t_e$ of the edges which represents the probability edge $e$ is selected at round $t$, i.e., $x^t_e = \P{e \in p^t}$.} The major challenge now is to ensure that each resource is sufficiently explored. We resolve the exploration problem by introducing the notion of \textit{Bounded-Away Description Polytope} (Definition~\ref{l:path_polytope2}) which guarantees that the selection probability of each useful resource is greater than an exploration parameter $\mu > 0$. The only similar idea in the literature we are aware of comes from \cite{chen2021impossible} that used it in the context of online predictions with experts advice.

We proceed with some necessary definitions and two important characterization lemmas.
\begin{definition}\label{d:active_edges} The set of active resources for agent $i \in [n]$ is the set
$E_i := \{e \in E:~e \in p_i \text{ for some } p_i \in \mathcal{P}_i \}$.
\end{definition}
\begin{lemma}
\label{lemma:active}
%\begin{restatable}{lemma}{active} \label{lemma:active}
Given the implicit description $\mathcal{X}_i$ of the strategy space $\mathcal{P}_i$, the set of active resources $E_i$ can be computed in polynomial-time.  %\end{restatable}
 \end{lemma}
\begin{definition}\label{l:path_polytope2} The $\mu$-\textit{Bounded-Away Description Polytope} for an exploration parameter $\mu >0$ is defined as,
\begin{align*}
    \mathcal{X}^{\mu}_i \triangleq \bigg\{x \in  \mathcal{X}_i :~x_e \geq \mu \quad \forall e \in E_i \bigg\}
\end{align*} 
\end{definition}
\begin{restatable}{lemma}{lemmanonempty} \label{lemma:non_empty}
The set $\mathcal{X}^{\mu}_i$ is non-empty for all $\mu \leq 1 / |E_i|$.
\end{restatable}

Notice that if a point $x_i \in \mathcal{X}^{\mu}_i$ then $x_i \in \mathcal{X}_i$. Since $\mathcal{X}_i = \mathrm{conv}(\mathcal{\hat{P}}_i)$ then any point  $x_i \in \mathcal{X}^{\mu}_i$ can be decomposed to a probability distribution over strategies $p_i \in \mathcal{P}_i$.
\begin{theorem}[Carathéodory Decomposition]\label{l:caratheodory}
For any point $x\in \mathcal{X}_i$ there exists a probability distribution $\pi_{x} \in \Delta(\mathcal{P}_i)$ with support at most $m+1$ strategies $p_i$ in $\mathcal{P}_i$ such that for all edges $e \in E$,
\[x_e = \sum_{p_i:e \in p_i} \mathrm{Pr}_{\pi_{x}}\left[p_i \text{ is selected} \right].\]
Such a distribution (it may not be unique) $\pi_{x}$ is called a \textit{Carathéodory decomposition} of $x$.
\end{theorem}
A Carathéodory decomposition of a point $x \in \mathcal{X}^{\mu}_i$ can be computed in polynomial-time through the \textit{decomposition algorithm} described in \cite{GLS88}.
\begin{theorem}\cite{GLS88}
Let a polytope $\mathcal{X}:=\{x\in \mathbb{R}^m:~A\cdot x \leq b\}$ with $A \in \mathbb{R}^{r \times m}$ and $b \in \mathbb{R}^r$. The Carathéodory Decomposition of any point $x \in \mathcal{X}$ can be computed in polynomial-time with respect to $r$  and $m$.
\end{theorem}
For the important special case of \textit{path polytopes} described in Section~\ref{sec:preliminaries}, a point $x \in \mathcal{X}_i^\mu$ can be decomposed to a probability distribution over $(s_i,t_i)$-paths with a simple and efficient algorithm outlined in \Cref{alg:charateodory}.% In Section~\ref{sec:experiments} we use this algorithm to compute the Caratheody decompositions in the context of network congestion games.
\begin{algorithm}[h]
			\caption{Efficient Computation of Charatheodory decomposition for Path Polytopes}
			\label{alg:charateodory}
			\begin{algorithmic}[1]
				\STATE {\bfseries Input:} A point $x \in \mathcal{X}^{\mu}_i$
                  \smallskip
                  \STATE $\mathrm{Res} \leftarrow \varnothing$
				\WHILE{$\sum_{e \in \mathrm{Out}(s_i)}x_e > 0$}
                    \smallskip
                    \STATE Let $A \leftarrow E \cap \{e\in E:~x_e >0\}$
                    \smallskip
			    \STATE Let $e_{\min} \leftarrow \argmin_{e \in A} x_e$ and $x_{\min} \leftarrow x_{e_{\min}}$.	    \smallskip
			    \STATE $\hat{p}_i \leftarrow$ An $(s_i,t_i)$-path of $G(V,A)$ with $e_{\text{min}} \in \hat{p}_i$.
                \smallskip
                \STATE $x_e \leftarrow x_e - x_{\text{min}}$ for all $e \in \hat{p}_i$.
                \smallskip
                \STATE $\mathrm{Res} \leftarrow \mathrm{Res} \cup \{\left(\hat{p}_i,x_{\text{min}}\right)\}$.
 
				\ENDWHILE
				\STATE {\bfseries return} $\mathrm{Res}$
			\end{algorithmic}
		\end{algorithm}

\begin{restatable}{lemma}{lemmadecompose}
\label{l:decompose_lemma}
Algorithm~\ref{alg:charateodory} requires $\mathcal{O}\left(|V||E|+|E|^2\right)$ steps to give a Caratheodory Decomposition for path polytopes.
\end{restatable}
In our experimental evaluations for network congestion games, presented in Section~\ref{sec:experiments}, we use Algorithm~\ref{alg:charateodory} to more efficiently implement \textit{Online Gradient Descent with Caratheodory Exploration} that we subsequently present.  
\subsection{Our Algorithm and Formal Guarantees}
In this section we present our semi-bandit online learning algorithm called 
\textit{Online Gradient Descent with Caratheodory Exploration} (Algorithm~\ref{alg:alg1}) for Online Resource Selection.

\begin{algorithm}[h]
			\caption{Online Gradient Descent with Caratheodory Exploration for Agent $i$}
			\label{alg:alg1}\nonumber
			\begin{algorithmic}[1]
                    \STATE $\mu_1 \leftarrow 1/|E_i|$
                    \smallskip
				\STATE Agent $i$ selects an arbitrary $x^1_i \in \mathcal{X}^{\mu_1}_i$.
    \smallskip
				\FOR{each round $t=1,\ldots, T$}
       \smallskip
       \STATE Agent $i$ computes a Caratheodory decomposition $\pi^t_i \in \Delta(\mathcal{P}_i)$ for $x^t_i  \in \mathcal{X}^{\mu_t}_i$.	
       \smallskip
       \STATE Agent $i$ samples a strategy $p^t_i \sim \pi^t_i$ and suffers cost, \[C_i^t(p_i^t,c^t):=\sum_{e\in p^t_i}c_e^t.\] 
       \textcolor{blue}{/* $c_e^t$ is the cost of edge $e$ with load the number of agents that chose $e$ at time $t$. */}
\smallskip       
         \STATE Agent $i$ sets $\hat{c}^t_e \leftarrow c^t_e\cdot \mathds{1}\left[e \in p_i^t\right]/x^t_e$ for all $e \in E$.
\smallskip
	\STATE Agent~$i$ updates $x^{t+1}_i \in \mathcal{X}^{\mu}_i$ as,
       \[x_i^{t+1} = \Pi_{\mathcal{X}^{\mu_{t+1}}_i}\left[x_i^t - \gamma_t \cdot \hat{c}^t \right],\]
        where  $\gamma_t \leftarrow t^{-3/5}$ and $\mu_t \leftarrow \min(1/m_i, t^{-1/5}).$
				\ENDFOR
			\end{algorithmic}
		\end{algorithm}

At Step~$4$, $\mathrm{OGD}-\mathrm{CE}$ performs a Caratheodory Decomposition to convert the fractional point $x_i^t \in \mathcal{X}^{\mu_t}_i$ into a probability distribution $\pi^t_i$ over pure strategies $p_i \in \mathcal{P}_i$. The latter guarantees that the experienced cost equals the fractional congestion cost, i.e.,
\begin{equation}\label{eq:imp}
\E \left[ \sum_{e \in p^t} c_e^t~|~ x_i^t \right] =  \innerprod{c^t}{x_i^t}.
\end{equation}
At Step~$7$, $\mathrm{OGD}-\mathrm{CE}$ runs a step of Online Gradient Descent to the \textit{time-expanding polytope} $\mathcal{X}^{\mu_t}_i$ that approaches $\mathcal{X}_i$ as $t \rightarrow \infty$. The latter is crucial as it gives the following:
\begin{lemma} \label{lemma:properties_estimator_online} The estimator $\hat{c}^t_e = c^t_e\cdot \mathds{1}\left[e \in p^t\right]/x^t_e$ satisfies
\begin{enumerate}
\item $\E\left[\hat{c}^t_e \right] = c^t_e$ for $e \in E_i$~~~~~~~~~ \textit{(Unbiasness)}.
\item $|\hat{c}^t_e| \leq c_{max}/\mu_t$ for $e \in E_i$~~~\textit{(Boundness)}.
\end{enumerate}
\end{lemma}
\begin{remark}
Projecting to the time-expanding polytope $\mathcal{X}^{\mu_t}_i$ and not to $\mathcal{X}_i$ is crucial since in the latter case we cannot control the variance of the estimator, as $x_e^t$ can go to $0$ arbitrarily fast. On the other hand, it is crucial to 
compute a Caratheodory Decomposition with respect to $\mathcal{X}_i$ and not with respect to $\mathcal{X}^{\mu_t}_i$ since the extreme points of 
$\mathcal{X}^{\mu_t}_i$ do not correspond to pure strategies $p_i \in \mathcal{P}_i$.
\end{remark}
 We conclude the section by presenting the formal guarantees of Algorithm~\ref{alg:alg1}. In Theorem~\ref{t:regret} we establish the \textit{no-regret} property of Algorithm~\ref{alg:alg1}. In Theorem~\ref{thm:convergence_to_nash} and Corollary~\ref{c:markov} we present its convergence guarantees.

\begin{restatable}{theorem}
{thmnoregret}\label{t:regret}
Let $p_i^1,\ldots,p^T_i \in \mathcal{P}_i$ the sequence of strategies produced by Algorithm~\ref{alg:alg1} given as input the costs $c^1,\ldots,c^T$ with $\norm{c^t}_\infty \leq c_{\max}$. Then with probability $1- \delta$, 
\[\sum_{t=1}^T\sum_{e \in p^t_i} c_e^t - \min_{p_i^\star \in \mathcal{P}_i}\sum_{e \in p^\star_i} c_e^t \leq \mathcal{O}\left(mc^2_{\max}T^{4/5} \log(1/\delta) \right)\]
\end{restatable}

\begin{remark}
\label{remarkT34}
Setting $\delta := \mathcal{O}\left(1/mTc_{\max}\right)$ directly implies that Algorithm~\ref{alg:alg1} admits regret $\tilde{\mathcal{O}}\left(m^2 c_{\max}^2 T^{4/5}\right)$ regret. The better $\mathcal{O}(m^{3/2}T^{3/4})$ regret bound can be obtained by selecting the parameters $\gamma_t = c_{\max}^{-1}m^{-1/2}t^{-3/4}$ and $\mu_t = \min(1/m_i, m^{-1/2}t^{-1/4})$. However such a parameter selection leads to $\mathcal{O}(T^{-1/8})$ convergence rate to NE. 
\end{remark}
In Theorem~\ref{thm:convergence_to_nash} we establish that the agents converge to a NE if all agents adopt Algorithm~\ref{alg:alg1}.

\begin{restatable}{theorem}{thmconvergencetoNash}
\label{thm:convergence_to_nash}
Let $\pi^1,\ldots, \pi^T$ the sequence of strategy profiles produced if all agents adopt Algorithm~\ref{alg:alg1}. Then for all $T \geq \Theta\left( m^{12.5}n^{7.5}/\epsilon^5 \right)$,
\[\frac{1}{T}\E \left[\sum_{t=1}^T\max_{i \in [n]}\left[c_i(\pi^{t}_i, \pi^{t}_{-i}) - \min_{\pi_i \in \Delta(\mathcal{P}_i)} c_i(\pi_i, \pi_{-i}^{t})\right]\right] \leq \epsilon.\]
The same holds for $T \geq \Theta(n^{6.5}m^{7}/\epsilon^5)$ in case the agents know $n,m$ and select $\gamma_t :=\Theta(m^{-4/5} n^{-8/5}c^{-1}_{\max} t^{-3/5})$ and $\mu_t := \Theta(n^{-6/5}m^{-11/10} t^{-1/5})$.
\end{restatable}
We note that the exact same notion of \textit{best-iterate convergence} (as in Theorem~\ref{thm:convergence_to_nash}) is considered in \cite{du22,leonardos2022global,DWZJ22,Ana22}. In Corollary~\ref{c:markov} we present a more clear interpretation of Theorem~\ref{thm:convergence_to_nash}.
\begin{restatable}{corollary}{Markov}
\label{c:markov}
In case all agents adopt Algorithm~\ref{alg:alg1} for $T\geq \Theta(n^{6.5}m^{7}/\epsilon^5)$ (resp. $\Theta\left( m^{12.5}n^{7.5}/\epsilon^5 \right)$) then with probability $\geq 1-\delta$, 
\begin{itemize}
\item $(1-\delta)T$ of the strategy profiles $\pi^1,\ldots, \pi^T$ are $\epsilon/\delta^2$-approximate Mixed NE.
\item $\pi^{t}$ is an $\epsilon/\delta$-approximate Mixed NE once $t$ is sampled uniformly at random in $\{1,\ldots,T\}$.
\end{itemize}
\end{restatable}

%\begin{restatable}{corollary}{Markov}
%\label{c:markov}
%If all agents adopt Algorithm~\ref{alg:alg1} for $T\geq \Theta(n^{6.5}m^{7}/\epsilon^5)$ (resp. $\Theta\left( m^{12.5}n^{7.5}/\epsilon^5 \right)$) then with probability $\geq 1-\delta$, $(1-\delta)$ fraction of the profiles $\pi^1,\ldots, \pi^T$ are $(\epsilon/\delta^2)$-Mixed NE. Moreover let $t$ be a uniformly sampled over $\{1,T\}$ then with probability $\geq 1-\delta$, $\pi^{t}$ is an $\epsilon/\delta$-Mixed NE.
%\end{restatable}

%By Theorem~\ref{thm:convergence_to_nash1} it directly follows the following Lemma. 

%\begin{restatable}{theorem}{thmconvergencetoNash}
%\label{thm:convergence_to_nash}
%Let $\pi^1,\ldots, \pi^T \in \Delta(\mathcal{P}_1)\times\ldots\times \Delta(\mathcal{P}_n)$ the sequence of strategy profiles produced in case each agent $i \in [n]$ adopts Algorithm~\ref{alg:alg1}. Then with probability $1-\delta,$ at least $(1-\delta)T$ of the rounds $t^\ast \in \{1,...,T\}$ give strategy profiles $\pi^{t^\ast}$ that are  $\mathcal{\tilde{O}}\br{\sqrt{m^{5} n^{3}c^{3}_{\max}}T^{-1/5}/\delta^2}$-MNE. Moreover in case $\gamma_t =\mathcal{O}(m^{-4/5} n^{-8/5}c^{-1}_{\max} t^{-3/5})$ and $\mu_t = \mathcal{O}(n^{-6/5}m^{-11/10} t^{-1/5})$ in Algorithm~\ref{alg:alg1} then an $\epsilon$-MNE is reached in $\mathcal{O}(n^{6.5}m^{7}/\epsilon^5)$ steps.
%\end{restatable}

\section{Sketch of Proof of Theorem~\ref{t:regret}}
\label{sec:first_thm_sketch}
In this section, we present the basic steps of the proof of Theorem~\ref{t:regret}. Due to Equation~\ref{eq:imp} established by the Caratheodory decomposition and the fact that $|\hat{c}_e^t| \leq c_{\mathrm{max}}/\mu_t$ we establish the following concentration result for the quantity 
$\sum_{t=1}^T\left[ \sum_{e \in p^t} c_e^t  - \innerprod{c^t}{x_i^t}\right]$.
\begin{restatable}{lemma}{thmconcentration}\label{thm:concentration}
Let the sequences $x^1_i,\ldots,x^T_i \in \mathcal{X}_i$ and $p^1_i,\ldots,p^T_i \in \mathcal{P}_i$ produced by \Cref{def:path_selection_protocol}. Then, with probabilty $1-\delta/2$,
\[ \sum_{t=1}^T \sum_{e \in p^t_i} c_e^t  = \sum_{t=1}^T\innerprod{c^t}{x_i^t} + \mathcal{O}\br{\sqrt{c_{\max}}\log (m/\delta)\sqrt{T}} .\]
\end{restatable}
Let $p_i^\ast \in \mathcal{P}_i$ denote the optimal strategy for the sequence of costs $c^1,\ldots,c^T$ and $x^\star_i \in \mathcal{X}_i$ the corresponding $\{0,1\}^{m}$ extreme point of $\mathcal{X}_i$. Then Lemma~\ref{thm:concentration} implies that
with probability $1 - \delta/2$,
\begin{equation*}\label{eq:1}
     \sum_{t=1}^T \left(\sum_{e \in p_i^t} c_e^t - \sum_{e \in p^\star_i} c_e^t\right)   = \sum_{t=1}^T \innerprod{c^t}{x_i^t - x^\star_i} + \tilde{\mathcal{O}}\br{\sqrt{T}}.
     \end{equation*}
As a result, in the rest of the section we bound the term $\sum_{t=1}^T \innerprod{c^t}{x_i^t - x^\star_i}$. %To do so we introduce the following definition.

Unfortunately, this term can not be directly control because at any step $t$ the comparator point $x^\star_i$ is not necessarly in $\mathcal{X}^{\mu_t}$. To overcome the issue we construct a sequence of comparator points $\bc{x^\star_{\mu_t}}^T_{t=1}$ that are guaranteed to satisfy $x^\star_{\mu_t}\in\mathcal{X}^{\mu_t}$.

Formally, we have the following definition.
\begin{definition}\label{d:project}
Let $p_i^\star \in \mathcal{P}_i$ the optimal strategy
for the sequence $c^1,\ldots,c^T$ and $x_i^\star \in \mathcal{X}_i$ its corresponding extreme point in $\mathcal{X}_i$ . Moreover, consider constructing a collection of strategies $\mathcal{D}=\{\tilde{p}^{\ell}\}^m_{\ell=1}$ sampled as follows. For any active resource $e\in E_i$, add to $\mathcal{D}$ a strategy $\tilde{p}\in\mathcal{P}_i$ such that $e\in\tilde{p}$. Considering the collected strategies in a vector form, i.e. elements of $\bc{0,1}^m$, we define
\[x^\star_{\mu_t} \triangleq (1 - m \mu_t)x^\star_i + \mu_t\sum^m_{\ell=1} \tilde{p}^{\ell}.\]

\begin{remark} \label{remark:comparator}
To see that $x^\star_{\mu_t}\in\mathcal{X}^{\mu_t}$, denote as $s$ the vector $s = \frac{1}{m}\sum^m_{i=1} \tilde{p}^{\ell}$. Since by construction, for any $e\in E_i$ there exists $\ell\in[m]$ such that $e\in\tilde{p}^\ell$, we have that $s \geq 1/m$. Moreover, $s \in \mathcal{X}$ which implies  $s \in \mathcal{X}^{1/m}$. At this point, we can write $x^\star_{\mu_t} = (1 - m\mu_t)x_i^\star + m\mu_t s$. Then, it is evident that $s\geq 1/m$ implies $x^\star_{\mu_t}\geq \mu_t$ and that $x^\star_{\mu_t}$ is a convex combination of $x^\star_i$ and $s$ because $\mu_t \leq \frac{1}{m}$ and therefore that $x^\star \in \mathcal{X}$. These two facts allow to conclude that $x^\star_{\mu_t} \in \mathcal{X}^{\mu_t}$.
\end{remark}

%Then the point 
%$x^\star_{\mu_t} \in \mathcal{X}_i^{\mu_t}$ is defined as,
%\[x^\star_{\mu_t} \triangleq \Pi_{\mathcal{X}^{\mu_t}_i}\left[x^\star_i\right].\]
\end{definition}

Up next, we decompose the right-hand term of Equation~\ref{eq:1} and separately bound each of the $(\mathrm{A} - \mathrm{C})$ terms.
\begin{align*}
\sum_{t=1}^T \innerprod{c^t}{x_i^t - x^\star} &= \underbrace{\sum_{t=1}^T \innerprod{\hat{c}^t}{x_i^t - x_{\mu_t}^\star}}_{(\mathrm{A})}  + \underbrace{\sum_{t=1}^T \innerprod{c^t}{x_{\mu_t}^\star - x^\star}}_{(\mathrm{B})}\\&+ \underbrace{\sum_{t=1}^T \innerprod{c^t - \hat{c}^t}{x_i^t - x_{\mu_t}^\star}}_{_{(\mathrm{C})}}.
\end{align*}
The bound on term $(\mathrm{A})$ is established in Lemma~\ref{thm:ogd}. 
\begin{restatable}{lemma}{thmogd}\label{thm:ogd}
Let the sequences $x^1_i,\ldots,x^T_i \in \mathcal{X}_i$ and $\hat{c}^1,\ldots,\hat{c}^T$ produced by \Cref{alg:alg1}. Then,
\begin{align*}\sum_{t=1}^T &\innerprod{\hat{c}^t}{x_i^t - x_{\mu_t}^\star} \leq \frac{2m}{\gamma_T} + mc_{\max}^2 \sum^T_{t=1}\frac{\gamma_t}{\mu_t^2}.
\end{align*}
\end{restatable}

The proof of Lemma~\ref{thm:ogd} is based on extending the arguments of \cite{Z07} for Online Projected Gradient Descent. The basic technical difficulty comes from the fact that in Step~$7$, Algorithm~\ref{alg:alg1} projects in the time-changing feasibility set $\mathcal{X}^{\mu_t}_i$ while in the analysis of \cite{Z07} the feasibility set is invariant.

The term~$(\mathrm{B})$ quantifies the suboptimality of the projections of $x^\star_i$ on the time expanding polytopes $\mathcal{X}_i^{\mu_t}$. Notice that $x^\star_i$ is possibly outside the $\mathcal{X}_i^{\mu_t}$ where \Cref{alg:alg1} projects to.
\begin{restatable}{lemma}{thmboundedpoly}(Sub-optimality of Bounded Polytopes)\label{thm:bounded_poly}
Let a sequence of costs $c^1,\ldots,c^T$
with $\norm{c^t}_\infty\leq c_{\max}$. Then,
\[\sum_{t=1}^T \innerprod{c^t}{x^\star_{\mu_t} - x^\star_i} \leq m^2 c_{\max} \sum^T_{t=1}\mu_t. \]
where $x_i^\star$ and  $x^\star_{\mu_t}$ are introduced in \Cref{d:project}.
\end{restatable}
Finally, we bound the term $(\mathrm{C})$ that quantifies the concentration of the cost estimators built at Step 7 of \Cref{alg:alg1} and the realized costs. The latter is established in Lemma~\ref{thm:concentration} and its proof lies on the \textit{Unbiasness} and \textit{Boundness} property of the estimator $\hat{c}_t$.
\begin{restatable}{lemma}{thmcostconcentration}\label{thm:cost_concentration}
Let $\hat{c}^1,\ldots,\hat{c}^T$ the sequence produced by Algorithm~\ref{alg:alg1} given as input the sequence of costs $c^1,\ldots,c^T$. Then with probability $1-\delta/2$,
\[ \sum_{t=1}^T \innerprod{c^t - \hat{c}^t}{x_i^t - x_{\mu_t}^\star} \leq \frac{m c_{\max}}{\mu_T}\sqrt{T \log(2m/\delta)}. \] 
\end{restatable}
\begin{figure*}[t]
\centering
\begin{tabular}{cc}

\subfloat[\label{subfig:regret2}$\frac{\mathcal{R}_{\mathcal{A}}(T)}{T}$]{%
    \includegraphics[width=0.42\linewidth]{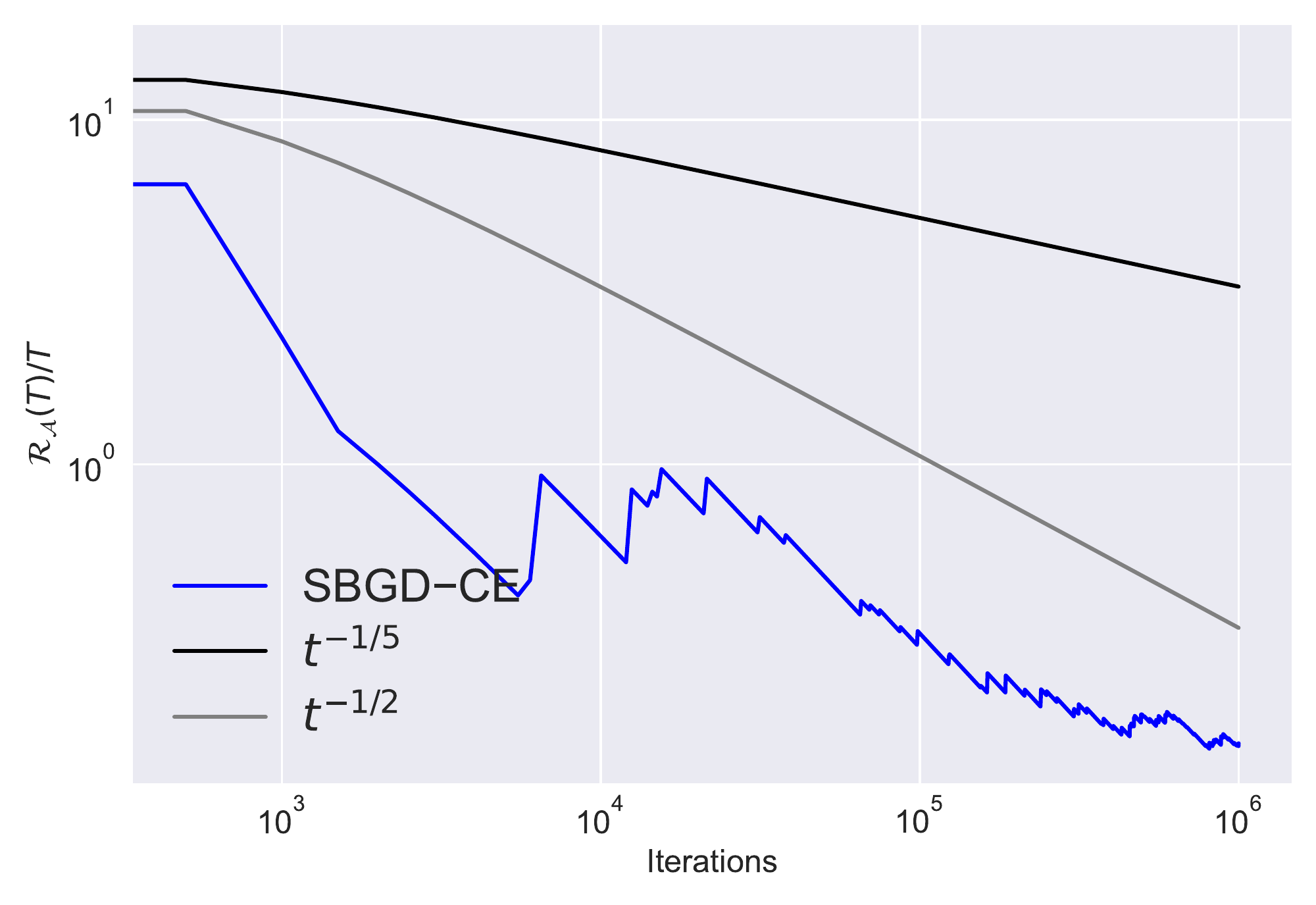}
     } 
    & \subfloat[\label{subfig:explo2}Exploitability]{%
    \includegraphics[width=0.42\linewidth]{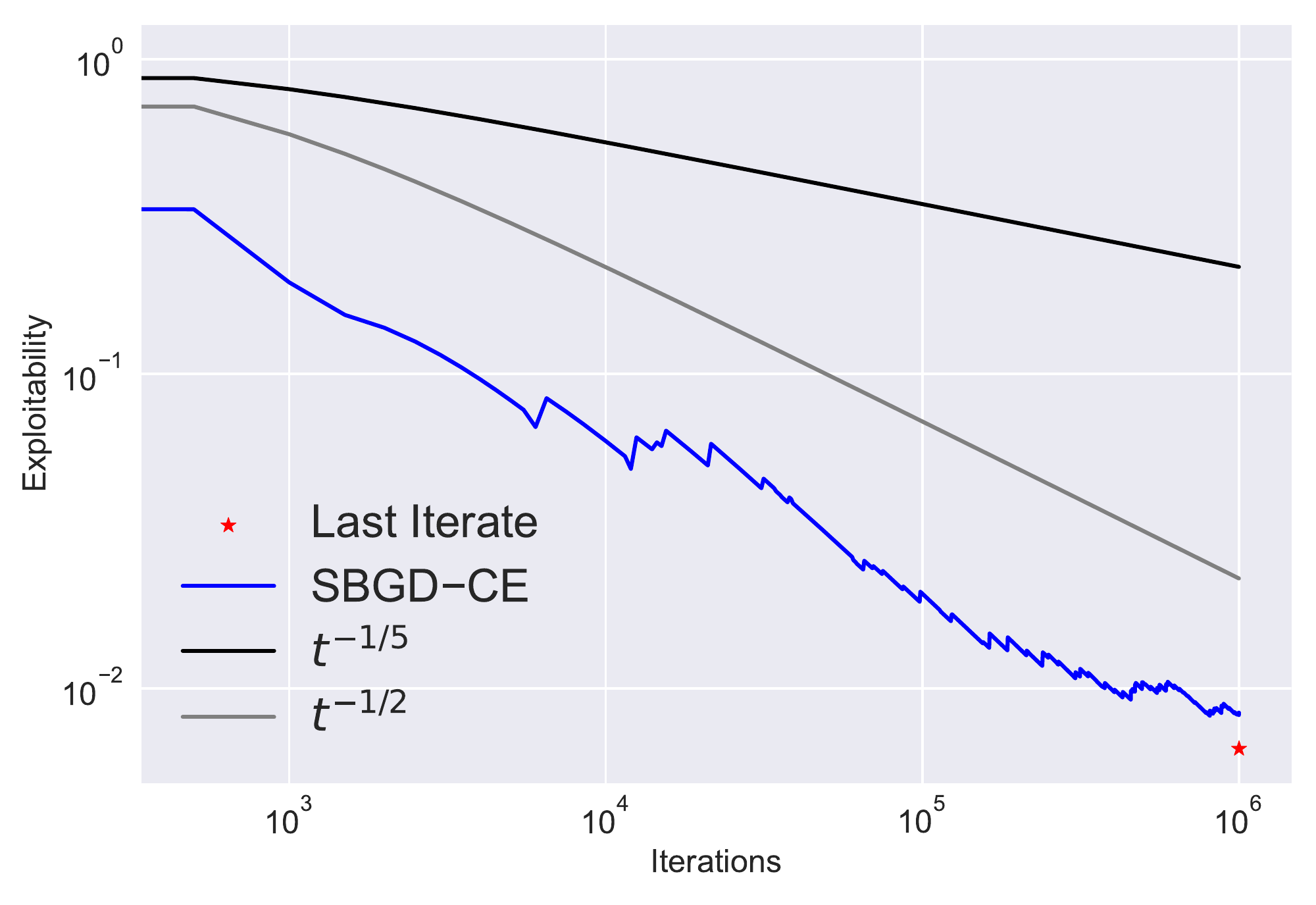}
     } \\
\end{tabular}
\caption{Regret and Exploitability on network games with $2$ agents.}
\label{fig:log2agent}
\end{figure*}
\section{Sketch of Proof of Theorem~\ref{thm:convergence_to_nash}}
\label{sec:second_thm_sketch}
In this section, we provide the main steps and ideas for proving Theorem~\ref{thm:convergence_to_nash}, establishing that in case all agents adopt Algorithm~\ref{alg:charateodory}, the overall system converges to a Mixed NE. We first introduce some important preliminary notions.

\begin{defn}[\textbf{Fractional Potential Function}]
\label{def:fractional_potential}
Let the fractional potential function $\Phi : \mathcal{X}_1 \times \ldots \times \mathcal{X}_n \rightarrow \mathbb{R}$ 
\begin{equation*}
\label{eq:lipschitz_potential}
        \Phi(x) \triangleq \sum_{e\in E} \sum_{\mathcal{S}\subset[n]} \prod_{j \in \mathcal{S}} x_{j,e} \prod_{j \notin \mathcal{S}} (1 - x_{j,e} ) \sum^{\abs{\mathcal{S}}}_{i=0} c_e(i).
\end{equation*}
\end{defn}
The potential function of Definition~\ref{def:fractional_potential} is crucial in our analysis since 
we can recast the problem of converging to NE into the problem of converging to a stationary point of $\Phi(x)$.
\begin{defn}%[$\epsilon$-stationary point]
\label{d:stationary}
A point $x = (x_1,\ldots,x_n) \in \mathcal{X}_1 \times \ldots \times \mathcal{X}_n$ is called an \textit{$(\epsilon, \mu)$-stationary point} $\Phi(x)$ if and only if 
\[\norm{x - \Pi_{\mathcal{X}^\mu} \left[x - \frac{1}{2 n^2  c_{\max} \sqrt{m}}  \nabla \Phi(x) \right]}\leq \epsilon\]
where $\mathcal{X}^\mu \triangleq \mathcal{X}^\mu_1 \times \cdots \times \mathcal{X}^\mu_n$.
\end{defn}
In Lemma~\ref{lemma:smoothness}, we establish that the potential function is smooth and uniformly bounded over its domain.
\begin{restatable}{lemma}{lemmasmoothness}
\label{lemma:smoothness}
The potential function $\Phi(\cdot)$ of Definition~\ref{def:fractional_potential} is smooth. More precisely,
\begin{equation*}
    \norm{\nabla \Phi (x) - \nabla \Phi (x')}_2 \leq 2 n^2  c_{\max} \sqrt{m} \cdot \norm{x - x'}_2
\end{equation*}
for all $x,x' \in \mathcal{X}_1\times \ldots \times \mathcal{X}_n $. Moreover, the potential function $\Phi(x)$ is bounded by $m n c_{\max}$. We also denote $\lambda \triangleq (2 n^2  c_{\max} \sqrt{m})^{-1}$ 
\end{restatable}
In \Cref{lemma:conversion} we formalize the link between approximate NE and approximate stationary points of the $\Phi(x)$.
\begin{restatable}{lemma}{lemmaconversion}
\label{lemma:conversion}
Let $\pi = (\pi_1,\ldots,\pi_n) \in \Delta(\mathcal{P}_1)\times \ldots \times \Delta(\mathcal{P}_n)$ and  $x = (x_1,\ldots,x_n) \in \mathcal{X}_1 \times \ldots \times \mathcal{X}_n$ such that for all resources $e \in E$,
\[x_{i,e} = \Pp{p_i \sim \pi_i}{e \in p_i}.\]
In case $x$ is an $(\epsilon, \mu)$-stationary point of $\Phi(x)$ then $\mathbf{\pi}$ is a $(4 n^2  mc_{\max}\epsilon + 2 m^2 n c_{\max} \mu)$-approximate Mixed NE.
\end{restatable}
\subsection{Convergence to stationary points}
In this section, we show that in case all agents use Algorithm~\ref{alg:charateodory} to (randomly) select their strategies, the produced sequence $x^t = (x_1^t,\ldots,x_n^t)$ converges to a stationary point of the potential function $\Phi(x)$.

We first show that the updates generated by each agent's individual implementation of \Cref{alg:alg1} can be equivalently described as the update performed by stochastic gradient descent on the potential function projected on the \textit{time-varying polytope} $\mathcal{X}^{\mu_t}:= \mathcal{X}^{\mu_t}_1 \times \ldots \times \mathcal{X}^{\mu_t}_n$.

\begin{restatable}{theorem}{thmestimatorprops}
\label{thm:estimator_props}
If each agent $i$ (randomly) selects its strategy according to Algorithm~\ref{alg:charateodory}. Then the produced sequence of vectors $x^1,\ldots,x^T$ can be equivalently described as
\begin{equation} x^{t+1} = \Pi_{\mathcal{X}^{\mu_{t+1}}}\left[x^t - \gamma_t \cdot \nabla_t \right] \label{eq:update_rule}\end{equation}
where the estimator $\nabla_t \triangleq \bs{\hat{c}^t_1, \dots, \hat{c}^t_n }$ ($\hat{c}^t_i$ is the cost estimate generated by player $i$ according to Step 7 in \Cref{alg:alg1}) satisfies
\begin{enumerate}
\item $\E[ \nabla_t ] = \nabla \Phi(x^t)$ \textit{\;\; and \;\;2.} $\E\left[\norm{\nabla_t}^2 \right] \leq \frac{n c^2_{\max} m}{\mu_t}.$
\end{enumerate}
\end{restatable}

The main technical contribution of this section, is to establish that the sequence $x^1,\ldots,x^T$ produced by Equation~\ref{eq:update_rule} converges to an $(\epsilon, \mu)$-stationary point of Definition~\ref{d:stationary}. The major challenge in proving the latter comes from the fact that in Equation~\ref{eq:update_rule} the projection step is respect to the time-changing polytope $\mathcal{X}^{\mu_t}$ while the projection in the definition of $(\epsilon, \mu)$-stationary point is with respect to the polytope $\mathcal{X}^\mu$.  
\begin{restatable}{theorem}{statpoint} \label{lemma:G_bound} Let $G(x) = \Pi_{\mathcal{X}^{\mu_T}}\bs{x - \lambda \nabla \Phi(x)} - x$ and the sequence $x^1,\ldots,x^T$ produced by Equation~\ref{eq:update_rule}. Then $\frac{1}{T}\sum^T_{t=1}\E \left[\norm{G(x^{t})}_2\right]$ is upper bounded by
\[ 2\sqrt{\frac{\lambda^2n mc_{\max} }{2 T \gamma_T}+ \frac{\lambda c^2_{\max} n m \sum_{t=1}^T \frac{\gamma_t^2}{\mu_t}}{ 2 T \gamma_T}}
    + \frac{8\sqrt{n  m^3} }{T} \sum^T_{t=1} \mu_{t}.\]
\end{restatable}

%At this point, we notice that the choices for $\{\mu_t\}_t$ and $\{\gamma_t \}_t$ given in \Cref{thm:convergence_to_nash} guarantees convergence to an approximate Nash equilibrium. }
\begin{figure*}[!h] % "[t!]" placement specifier just for this example
\centering
\begin{tabular}{cc}
\subfloat[$\frac{\mathcal{R}_\mathcal{A}(T)}{T}$ \label{fig:logregret20}]{%
    \includegraphics[width=0.42\linewidth]{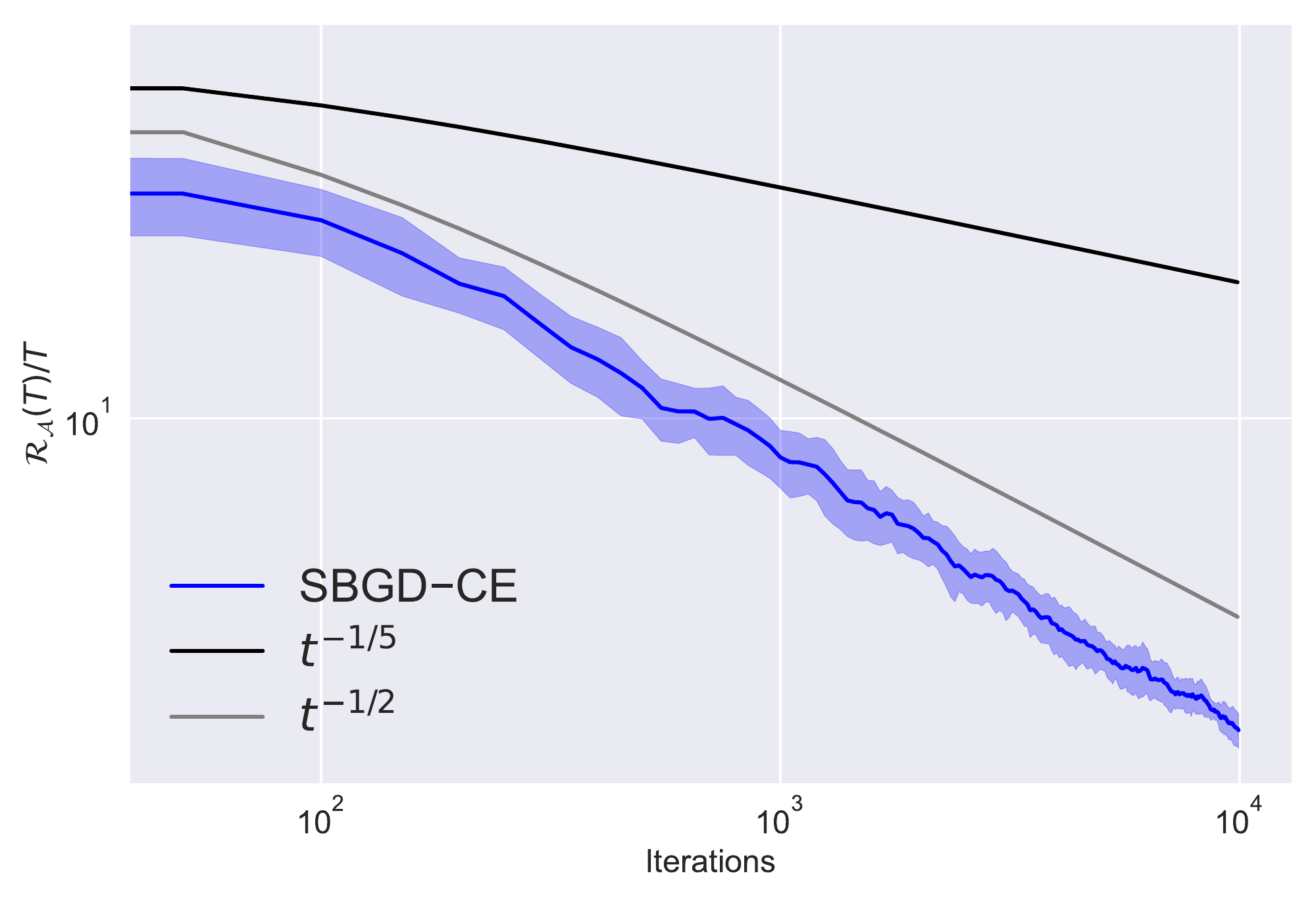}
     } &
\subfloat[Exploitability \label{fig:lognash20}]{%
    \includegraphics[width=0.42\linewidth]{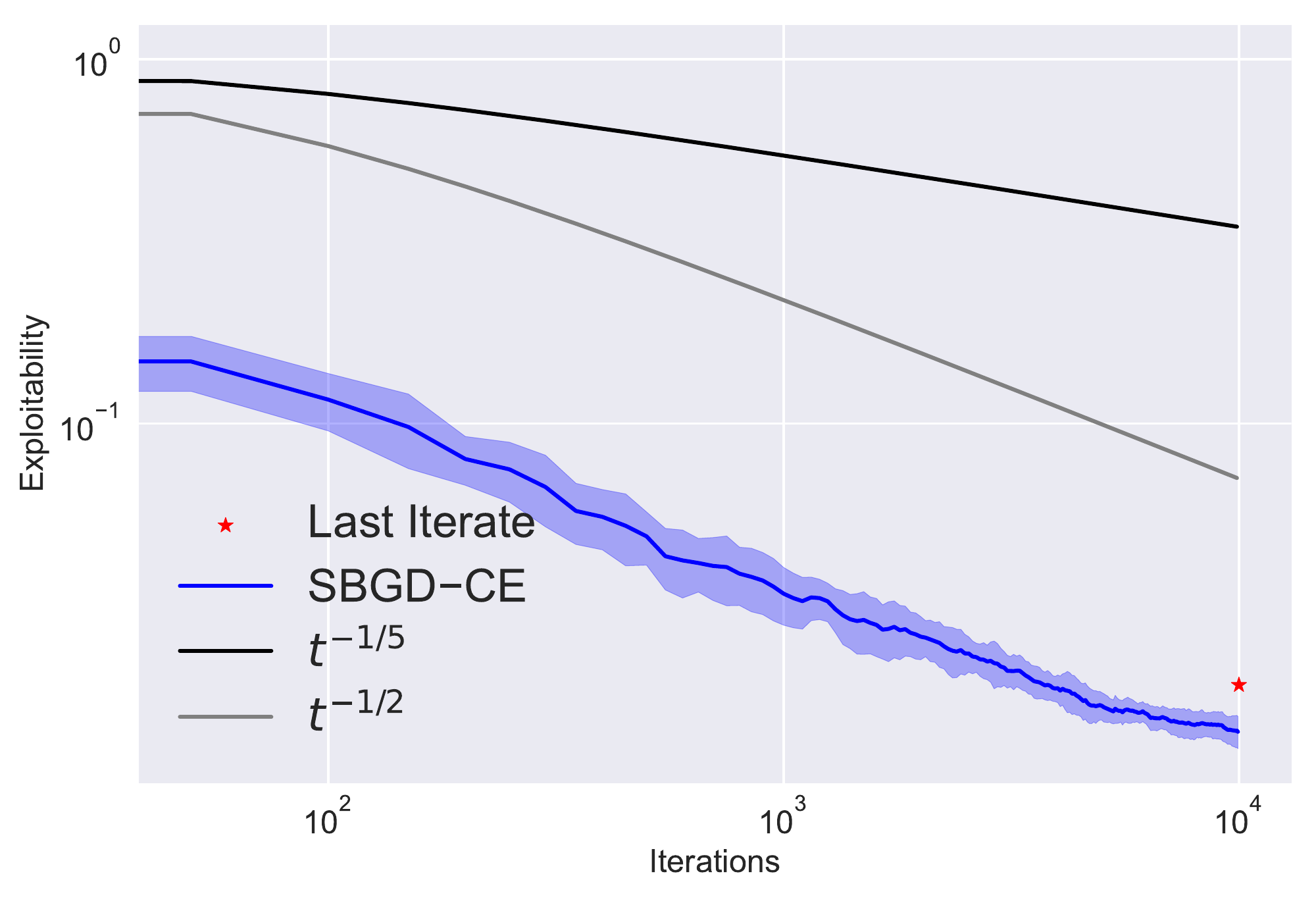}
     } \\
\subfloat[$\frac{\mathcal{R}_\mathcal{A}(T)}{T}$  \label{fig:logregret5}]{%
    \includegraphics[width=0.42\linewidth]{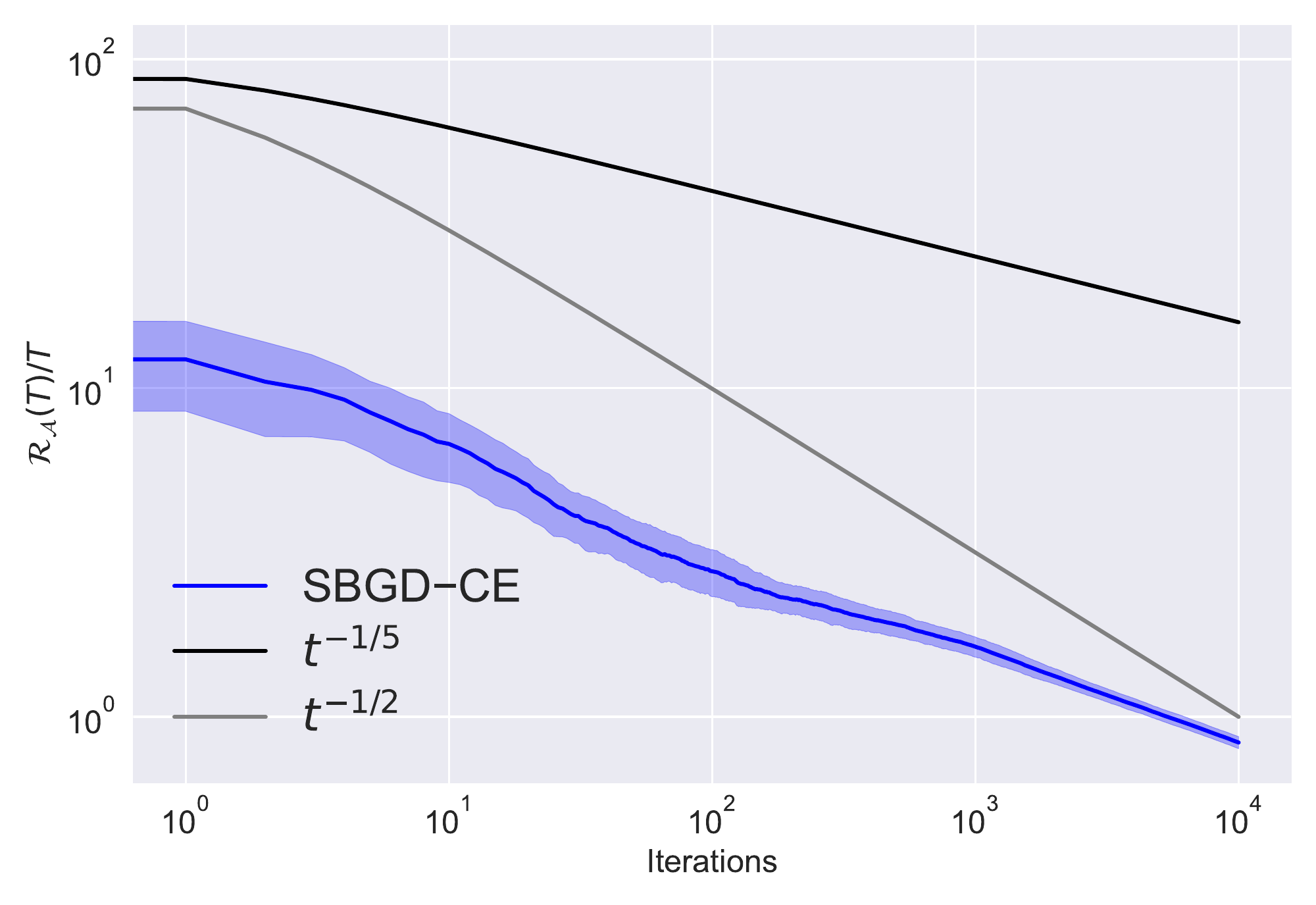}
     } &
\subfloat[Exploitability \label{fig:lognash_5_agent}]{%
    \includegraphics[width=0.42\linewidth]{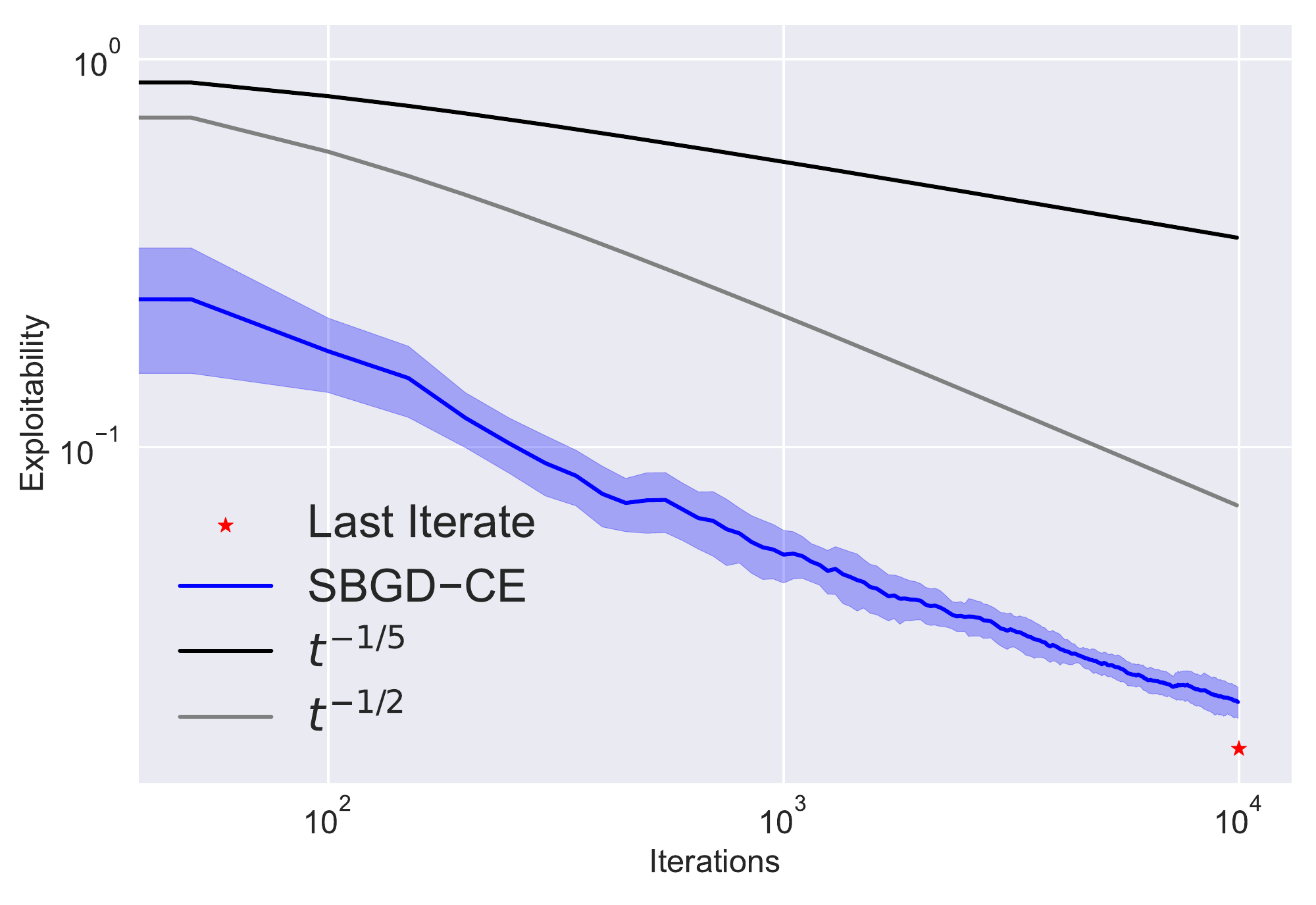}
     } \\
\end{tabular}
\caption{Experiments on network games with $20$ nodes for $20$ (\Cref{fig:logregret20} and \ref{fig:lognash20}) and $5$ agents (\Cref{fig:logregret5} and \ref{fig:lognash_5_agent}). Curves averaged over $10$ seeds for the 20 agents case and $50$ seeds for 5 agents.}
\label{fig:5agent_game}
\end{figure*}
\section{Experiments}
\label{sec:experiments}

We aim at verifying our theoretical statement by providing experiments in network congestion games. We consider a multigraph with chain topology composed by a set of nodes $\bc{v_i}^{|V|}_i$ (with $|V|=19$) where every node $v_i$ is connected only to $v_{i+1}$ by $2$ edges. Under this setting, Frank-Wolfe with Exploration \cite{du22} can not be implemented efficiently since there are $2^{19}$ possible paths. The same holds for \cite{leonardos2022global,DWZJ22}.
In order to verify empirically verify the convergence to NE, we monitor the exploitability of a strategy profile $(\pi_i,\pi_{-i})$ defined  as $\max_{i\in[n]} \frac{c_i(\pi_i,\pi_{-i}) -\min_{\pi^\prime_{i}} c_i(\pi'_i,\pi_{-i})}{c_i(\pi'_i,\pi_{-i})}$ which is $0$ for any NE. \Cref{subfig:explo2} shows the exploitability of the average path chosen by SBGD-CE. We notice it decreases at a rate $\approx t^{-1/2}$ which is better the theoretical bound provided in \Cref{thm:convergence_to_nash}. Furthermore, the red star in \Cref{subfig:explo2} represents the exploitability of the last iterate produced by \Cref{alg:alg1}, it can be seen that it also achieves a small value of exploitability.  We also verify the no-regret property of the algorithm in \Cref{subfig:regret2}. 
Experiments with $5$ and $20$ agents are provided in \Cref{fig:5agent_game}. The code is available at \url{https://github.com/lviano/SBGD-CE}.

\section{Conclusion}
This work introduces SBGD-CE which is the first no-regret online learning algorithm with semi-bandit feedback that once adopted by all agents in a congestion game converges to NE in the \textit{best-iterate sense}. As a result, our work answers an open question of \cite{du22} and improves upon their rates and complexity. The empirical evaluation inspires different future directions, in particular establishing \textit{last iterate} convergence rates to NE as well as tightening the rates for \textit{best-iterate convergence}.

\section*{Acknowledgements}
This work was supported by Hasler Foundation Program: Hasler Responsible AI (project number 21043), by the Swiss National Science Foundation (SNSF) under grant number 200021$\_$205011, by a PhD fellowship of the Swiss Data Science Center, a joint venture between EPFL and ETH Zurich, by Innosuisse under agreement number 100.960 IP-ICT. Ioannis Panageas would like to acknowledge UCI start-up grant. Xiao Wang acknowledges the Grant 202110458 from SUFE and support from the Shanghai Research Center for Data Science and Decision Technology.

\if 0
\begin{figure}[t!] % "[t!]" placement specifier just for this example
\centering
\begin{tabular}{cccc}
\subfloat[Uniform initialization]{%
    \includegraphics[width=0.24\linewidth]{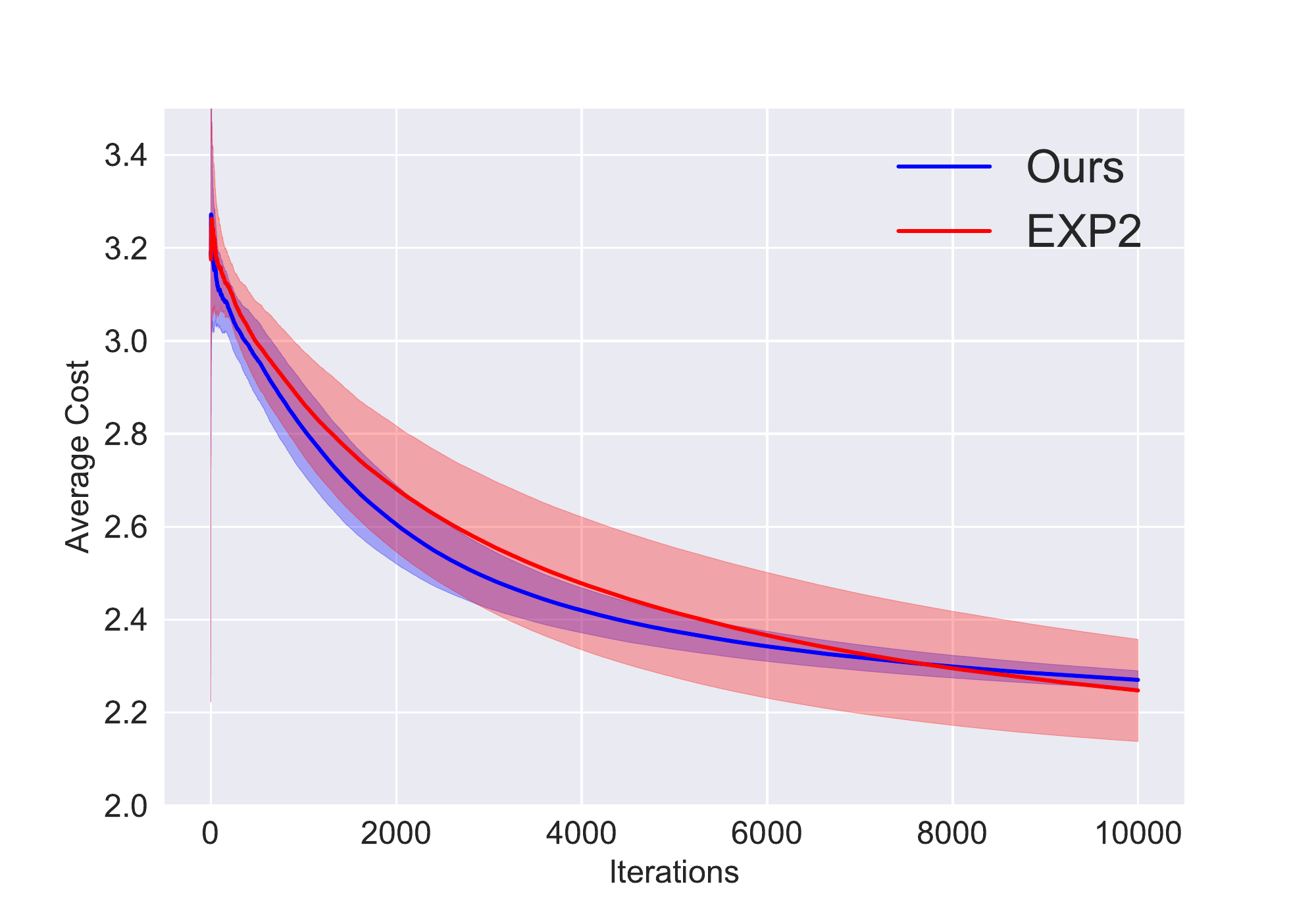}
     } &
\subfloat[High Traffic initialization]{%
    \includegraphics[width=0.24\linewidth]{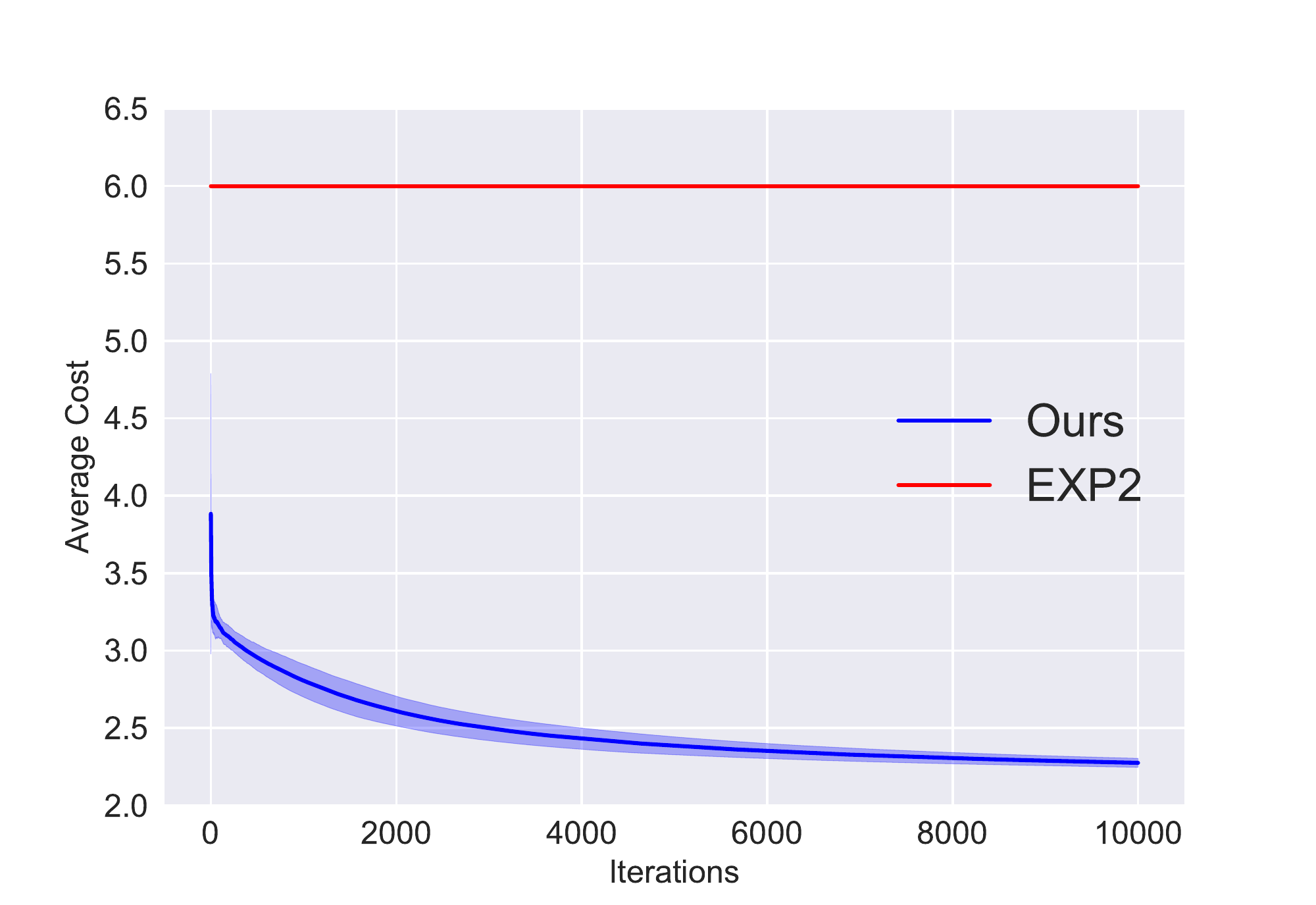}
     } \\
\end{tabular}
\caption{Experiments on network games with $3$ agents average of $50$ seeds.}
\label{fig:network_game}
\end{figure}
\fi

\bibliography{example_paper}
\bibliographystyle{plain}

\onecolumn
\appendix

\section{Additional related work on congestion games}
Congestion games, proposed in \cite{rosenthal73} are amongst the most well known and extensively studied class of games and have been successfully employed in myriad modeling problems. Congestion games have been proven to be isomorphic to potential games \cite{potgames}, and as a result, they always admit a potential function and a \textit{pure} Nash equilibrium. Moreover, (typically) due to existence of multiple Nash equilibria, Price of Anarchy has been proposed in \cite{KoutsoupiasP99WorstCE} for the purpose of efficiency guarantees in congestion games and is arguably amongst the most developed areas within algorithmic game theory, e.g.,
\cite{KoutsoupiasP99WorstCE,roughgarden2002bad,christodoulou,Fotakis2005226,Schafer10,Roughgarden09}. It is folklore knowledge that better-response dynamics in congestion games converge. In these dynamics, in every round, exactly one agent deviates to a better strategy. Convergence is guaranteed as the potential function always decreases along better response dynamics\footnote{Note that If two or more agents move at the same time then convergence is not guaranteed.}. Notwithstanding better dynamics converges, it has been shown that computing a pure Nash equilibrium is PLS-complete \cite{fabrikant04} and computing a (possible mixed) Nash equilibrium is CLS-complete \cite{babichenko2021settling}, i.e., it is unlikely to be able to provide an algorithm that computes (pure or mixed) Nash equilibrium in congestion games and runs in polynomial time (in the description of the game).
\section{Proof for \Cref{sec:preliminaries}}
\label{app:DAG}

\subsection{Path Polytope and Directed Acyclic Graphs}
\begin{definition}
A directed graph $G(V,E)$ is called \textit{acyclic} in case there are no cycles in $G(V,E)$.  
\end{definition}
\begin{definition}
Let a \textit{directed acyclic graph} $G(V,E)$ and the vertices $s_i,t_i \in V$. The $(s_,t_i)$-\textit{path polytope} is defined as follows, 
\begin{align*}
    \mathcal{X}_i \triangleq \bigg\{ \fcost\in\{0,1\}^m &: \sum_{e \in \mathrm{Out}(s_i)} x_e = 1\\ & \sum_{e \in\mathrm{In}(v)} x_e =  \sum_{e \in\mathrm{Out}(v)} x_e \quad \forall v \in V \setminus \bc{s_i, t_i} \\
    & \sum_{e \in \mathrm{In}(t_i)} x_e = 1 \bigg\} 
\end{align*}
\end{definition}

\begin{lemma}
The extreme points of the $(s_i,t_i)$-path polytope $\mathcal{X}_i$ correspond to $(s_i,t_i)$-paths of $G(V,E)$ and vice versa.
\end{lemma}
\begin{proof}
We first show that an $(s_i,t_i)$-path $p_i \in \mathcal{P}_i$ corresponds to an extreme point of $\mathcal{X}_i$.
Given an $(s_i,t_i)$-path $p \in \mathcal{P}_i$ consider the point $x_{p_i}$ of the polytope with $x_e^{p_i} = 1$ for all $e \in p_i$ and $x_e^{p_i} =0$ otherwise. Let assume that $x_{p_i}$ is not an extreme point of $\mathcal{X}_i$. Notice that $x_{p_i}$ satisfies all the constraints of $\mathcal{P}_i$ and since is a $\{0,1\}$-vector it is an extreme point of $\mathcal{X}_i$.

On the opposite direction we show that for any extreme point $x \in \mathcal{X}_i$ the set of edges $\{e \in E
:~x_e = 1\}$ is an $(s_i,t_i)$-path of $G(V,E)$. We first show that $x$ is necessarily integral ($x_e =0$ or $x_e = 1$). Let assume that there exists an edge $e\in E$ with $x_e \in (0,1)$ and consider $x_{\min} := \min_{e: x_e >0} x_e$. Consider an $(s_i,t_i)$-path containing edge $e_{\min}:= \argmin_{e: x_e >0} x_e$. Notice that for the point $x \in \mathcal{X}_0^i$ the following holds,
\[ x = x_{\min} \cdot p_{\min} + (1 - x_{\min}) \cdot \underbrace{\frac{x - p_{\min}}{1 - x_{\min}}}_{y}\]
Notice that $y \in \mathcal{X}_i$ and thus $x$ can be written as convex combination of $p_{\min}$ and $y$, meaning that $x$ cannot be an extreme point of $\mathcal{X}_i$. As a result, $x$ is an $\{0,1\}^m$-vector. Now consider the set of edges $p_x:=\{e \in E
:~x_e = 1\}$. Notice that node $s_i$ admits exactly one edge $e \in \mathrm{Out}(s_i)$ belonging in $p_x$. Similarly there exists exactly one edge $e \in \mathrm{In}(t_i)$ belonging in $p_x$. Due to the fact that $G(V,E)$ is acyclic, $p_x$ is necessarily an $(s_i,t_i)$-path.
\end{proof}

\section{Proofs of Section~\ref{s:algo_presentation}}

\subsection{Proof of \Cref{lemma:active}}
\begin{lemma}
%\begin{restatable}{lemma}{active} \label{lemma:active}
Given the implicit description $\mathcal{X}_i$ of the strategy space $\mathcal{P}_i$, the set of active resources $E_i$ can be computed in polynomial-time.  %\end{restatable}
 \end{lemma}
\begin{proof}
For each resource $e \in E$, consider the polytope $\mathcal{X}^e_i = \bc{ x \in \mathcal{X}_i : x_e =1}$ and we check whether it is empty. Notice that $\mathcal{X}^e_i$ admits $r_i + 1$ linear constraints and thus checking for its feasibility is done in polynomial time with respect to $r_i$ and $m$.

The correctness of the above algorithm can be established with the following simple argument. Let an $x \in\mathcal{X}_i^e$ then $x \in \mathcal{X}_i$ and additionally $x_e^i = 1$. The latter implies that $x$ can be decomposed to $m+1$ pure strategies $p_i \in \mathcal{P}_i$. Since $x_e^i=1$ any strategy $p_i$ participating in the convex combination of $x$ must admit $e \in p_i$. Thus $e \in E_i$. On the opposite direction, $\mathcal{X}_i^e$ cannot be empty in case $e \in E_i$. Notice that in the latter direction there exits $p_i \in \mathcal{P}_i$ with $e \in \mathcal{P}_i$ and thus the corresponding $\{0,1\}$-vector $x_{p_i} \in \mathcal{X}_i^e$.
\end{proof}
\subsection{Proof of \Cref{lemma:non_empty}}
\lemmanonempty*
\begin{proof}
Initialize $y =(0,\ldots,0) \in \{0,1\}^{m}$. For each  resource $e \in E_i$ select a strategy $p_i^e \in \mathcal{P}_i$ and update $y$ as $y \leftarrow y + \mu \cdot x_{p_i^e}$. The consider $x := y / |E_i|$. Notice that $y \in \mathcal{X}_i$ as convex combination
$x_{p_i^e} \in \mathcal{X}_i$. At the same time, $y_e \geq 1 / |E_i|$ for each $e \in E_i$. Thus, the set $\mathcal{X}_i^\mu$ is not empty for any $\mu \leq 1/|E_i|$.
\end{proof}
\lemmadecompose*
\begin{proof}
First, we notice that the algorithm always successful in DAGs because, in virtue of \Cref{lemma:path_existence}, we can always find a path in Step 6 of \Cref{alg:charateodory}.
At this point, in order to study the complexity of each iteration, we notice that Step 5 requires at most $\abs{E}$ read operations and  Step 6 can be performed in $\mathcal{O}(\abs{V} + \abs{E})$ operations.
Therefore, every iteration of the outermost loop require $\mathcal{O}(\abs{E} + \abs{V})$ operations.
Finally, the number iterations in bounded by $\abs{E}$ because Step 7 makes one coordinate of the point $x$ equal to $0$ at every iteration and this ensures that after at most $\abs{E}$ iterations $\sum_{e\in\mathrm{Out}(s)}x_e = 0$. 
This conclude the proof because the total operation complexity is $\mathcal{O}(\abs{E}(\abs{E} + \abs{V}))$ as stated in the main text.
\end{proof}

\subsection{Proof of \Cref{lemma:properties_estimator_online}}
For proving that $\hat{c}^t_e$ is unbiased, we have to recall that $x^t_e = \E\bs{\mathds{1}\bs{e \in p^t}}$ from which it follows that
\begin{align*}
\E\bs{\hat{c}^t_e} = \E\bs{\frac{c^t_e}{x^t_e}\mathds{1}\bs{e \in p^t}} = \frac{c^t_e}{x^t_e}\E\bs{\mathds{1}\bs{e \in p^t}} = c^t_e
\end{align*}
The second part of the proof concerns bounding the absolute value of the cost estimate. More precisely, we show that $|\hat{c}^t_e| \leq c_{max}/\mu_t$ for all $ e \in E_i$. Indeed,
\begin{align*}
    |\hat{c}^t_e| = \abs{\frac{c^t_e}{x^t_e}\mathds{1}\bs{e \in p^t}} \leq \frac{c_{\max}}{\mu}
\end{align*}

\section{Proofs of \Cref{sec:first_thm_sketch}}
\subsection{Proof of~\Cref{thm:concentration}}
\thmconcentration*
\begin{proof} 
Since the objective is a linear function $\min_{\fcost\in\mathcal{X}} \sum^T_{t=1}  \innerprod{\cost^t}{\fcost} =  \min_{p\in\mathcal{P}} \sum^T_{t=1} \sum_{e \in p} c^t_e$, then  we  have that
\begin{equation*}
    \sum^T_{t=1} X_t = \sum^T_{t=1} \sum_{e\in E} c_e^t \cdot  \left(\E_t  \left[\mathds{1}[e \in p^t]\right] - \mathds{1}[e \in p^t]\right)
\end{equation*}
where $\E_t$ is an expectation over the choice of $p^t$ conditioned on the filtration adapted to the process $(p^1, c^1, \dots, p^{t-1}, c^{t-1})$. As $c^t$ is conditionally independent on $p^t$, we get that
\begin{equation*}
    \sum^T_{t=1} X_t  = \sum_{e\in E}\sum^T_{t=1}   \left(\E_t\left[ c_e^t\mathds{1}[e \in p^t] \right] - c_e^t\mathds{1}[e \in p^t]\right)
\end{equation*}
Hence, we recognize the martingale difference sequence $\left(\br{\E_t \left[ c_e^t\mathds{1}[e \in p^t]\right] - c_e^t\mathds{1}[e \in p^t]}\right)^T_{t=0}$ that satisfies $\abs{\E_t \left[ c_e^t\mathds{1}[e \in p^t]\right] - c_e^t\mathds{1}[e \in p^t]} \leq c_{\max}$ therefore by Azuma-Hoeffding inequality we can conclude that with probability $1 - \delta_1$ for every resource $e \in E$
\begin{equation*}
    \sum^T_{t=1}   \br{\E_t\left[ c_e^t\mathds{1}[e \in p^t]\right] - c_e^t\mathds{1}[e \in p^t]} \leq \sqrt{\frac{1}{2}c_{\max} \log \frac{m}{\delta_1} T}.
\end{equation*}    
\end{proof}

\thmogd*
\begin{proof}

\begin{align*}
    \innerprod{\hat{\cost}^t}{x_i^t - x_{\mu_t}^\star} &\leq \frac{\innerprod{x_i^t - x_i^{t+1}}{x_i^t - x_{\mu_t}^\star}}{\gamma_t} \\&= \frac{1}{2\gamma_t} \br{\norm{x_{\mu_t}^\star - x_i^t}^2 - \norm{x_{\mu_t}^\star - x_i^{t+1}}^2 + \norm{x_i^{t+1} - x_i^t}^2}
    \\&\leq \frac{1}{2\gamma_t} \br{\norm{x_{\mu_t}^\star - x_i^t}^2 - \norm{x_{\mu_t}^\star - x_i^{t+1}}^2} + \frac{\gamma_t}{2}\norm{\hat{c}^t}^2
    \\&\leq \frac{1}{2\gamma_t} \br{\norm{x_{\mu_t}^\star - x_i^t}^2 - \norm{x_{\mu_t}^\star - x_i^{t+1}}^2} + \frac{ c^2_{\max} m}{2}\frac{\gamma_t}{\mu^2_t}.
\end{align*}
where in the first and third equality we use the contraction property of the projection, in the second equality we developed the square and in the last inequality we the bound on the norm of the estimator.
Summing over $t$ we obtain
\begin{align*}
    \sum^T_{t=1}\innerprod{\cost^t}{x_i^t - x_{\mu_t}^\star}
    &\leq \sum^T_{t=1}\frac{1}{2\gamma_t} \br{\norm{x_{\mu_t}^\star - x_i^t}^2 - \norm{x_{\mu_t}^\star - x_i^{t+1}}^2} + \frac{ c^2_{\max} m}{2}\sum^T_{t=1}\frac{\gamma_t}{\mu^2_t} \\ &\leq \sum^T_{t=1}\br{\frac{1}{2\gamma_t} \norm{x_{\mu_t}^\star - x_i^t}^2 - \frac{1}{2\gamma_{t+1}} \norm{x_{\mu_{t+1}}^\star - x_i^{t+1}}^2} \\&\phantom{=}+ \sum^T_{t=1}\br{\frac{\norm{x_{\mu_{t+1}}^\star - x_i^{t+1}}^2 }{2\gamma_{t+1}} - \frac{\norm{x_{\mu_t}^\star - x_i^{t+1}}^2 }{2\gamma_{t}}} + \frac{ c^2_{\max} m}{2}\sum^T_{t=1}\frac{\gamma_t}{\mu^2_t}
    \\ &\leq \frac{\norm{x_{\mu_T}^\star - x_i^T}^2}{2\gamma_T}  - \frac{\norm{x_{\mu_0}^\star - x_i^{1}}^2} {2\gamma_{1}} + \frac{3 m}{2 \gamma_T} + \frac{ c^2_{\max} m}{2}\sum^T_{t=1}\frac{\gamma_t}{\mu^2_t} 
    \\ &\leq \frac{2m}{\gamma_T}  + \frac{ c^2_{\max} m}{2}\sum^T_{t=1}\frac{\gamma_t}{\mu^2_t} 
\end{align*}
Where in the second last inequality we used
\begin{align*}
\sum^T_{t=1}\br{\frac{\norm{x_{\mu_{t+1}}^\star - x_i^{t+1}}^2 }{2\gamma_{t+1}} - \frac{\norm{x_{\mu_t}^\star - x_i^{t+1}}^2 }{2\gamma_{t}}} &= \sum^T_{t=1}\br{\frac{\norm{x_{\mu_{t+1}}^\star - x_i^{t+1}}^2 }{2\gamma_{t+1}} - \frac{\norm{x_{\mu_{t+1}}^\star - x_i^{t+1}}^2 }{2\gamma_{t}}} \\&\phantom{=}+ \sum^T_{t=1}\frac{\norm{x_{\mu_{t+1}}^\star - x_i^{t+1}}^2 - \norm{x_{\mu_t}^\star - x_i^{t+1}}^2 }{2\gamma_{t}} \\ &= \sum^T_{t=1}\br{\norm{x_{\mu_{t+1}}^\star - x_i^{t+1}}^2\br{\frac{1}{2\gamma_{t+1}} - \frac{1}{2\gamma_{t}}}}\\ &\phantom{=}+ \sum^T_{t=1}\frac{\innerprod{x^\star_{\mu_t} + x^\star_{\mu_{t+1}} - 2x_i^{t+1}}{x^\star_{\mu_{t+1}} - x^\star_{\mu_{t}}} }{2\gamma_{t}} \\ &\leq m\sum^T_{t=1}\br{\frac{1}{2\gamma_{t+1}} - \frac{1}{2\gamma_{t}}} + \sum^T_{t=1}\frac{\norm{x^\star_{\mu_t} + x^\star_{\mu_{t+1}} - 2x_i^{t+1}}\norm{x^\star_{\mu_{t+1}} - x^\star_{\mu_{t}}} }{2\gamma_{t}} \\ &\leq \frac{m}{2 \gamma_T} + \sqrt{m} \sum^T_{t=1}\frac{\norm{(1 - m\mu_t)x^\star_i + m\mu_t s - (1 - m\mu_{t+1})x^\star_i - m\mu_{t+1} s} }{\gamma_{t}} \\ &= \frac{m}{2 \gamma_T} + m^{3/2}\sum^T_{t=1}\frac{(\mu_t - \mu_{t+1})\norm{x^\star - s}}{\gamma_t}
\\ &\leq \frac{m}{2 \gamma_T} + m^2\sum^T_{t=1}\frac{(\mu_t - \mu_{t+1})}{\gamma_t}
\\ &\leq \frac{m}{2 \gamma_T} + \frac{m^2}{\gamma_T}\sum^T_{t=1}(\mu_t - \mu_{t+1}) \\ &\leq \frac{m}{2 \gamma_T} + \frac{m^2 \mu_1}{\gamma_T} \\ &\leq \frac{m}{2 \gamma_T} + \frac{m}{\gamma_T} = \frac{3m}{2\gamma_T}
\end{align*}
\if 0
Then, we recognize that the second term in right hand side of \cref{eq:decomposition} is a bounded martingale difference sequence. Indeed, it holds that
\begin{equation*}
\E_t\bs{\innerprod{\hat{\cost}^t - c^t}{x_{\mu_t}^\star - x_{\mu_T}^\star}} = 0
\end{equation*}
and 
\begin{equation*}
    \abs{\innerprod{\hat{\cost}^t - c^t}{x_{\mu_t}^\star - x_{\mu_T}^\star}} \leq \norm{\hat{\cost}^t - c^t}_{\infty}\norm{x_{\mu_t}^\star - x_{\mu_T}^\star}_{1} \leq c_{\max} (m + 1)(\mu_t + 1) \leq 4 c_{\max} m
\end{equation*}
where we used that $\norm{x_{\mu_t}^\star - x_{\mu_T}^\star}_{1} \leq \mu_t (m + 1)$ and $\norm{\hat{\cost}^t - c^t}_{\infty} \leq c_{\max} \br{1 + 1/\mu_t} $.
Then, by Azuma-Hoeffding inequality it holds that with probability $1 - \delta$
\begin{equation*}
    \sum^T_{t=1}{\innerprod{\hat{\cost}^t - c^t}{x_{\mu_t}^\star - x_{\mu_T}^\star}} \leq c_{\max} m \sqrt{32 T \log(1/\delta)}
\end{equation*}
 Finally for the third term in the right hand side of \Cref{eq:decomposition} we have that
 \begin{align*}
\sum^T_{t=1}\innerprod{c^t}{x_{\mu_t}^\star - x_{\mu_T}^\star} &\leq \sum^T_{t=1} \norm{c^t}_{\infty}{x_{\mu_t}^\star - x_{\mu_T}^\star}_{1} \\ &\leq c_{\max} (m + 1) \sum^T_{t=1} (\mu_t - \mu_T)
 \end{align*}
 Putting all together gives the result.
 \fi
\end{proof}

\thmboundedpoly*
\begin{proof}

\begin{align*}
    \sum^T_{t=1}  \innerprod{\cost^t}{\fcost^\star_{\mu_t} - \fcost^\star_i} & \leq \sum^T_{t=1} \norm{\cost^t}_{\infty} \norm{\fcost^\star_i - \fcost^\star_{\mu_t}}_1 \\ &
    \leq m^2 c_{\max}\sum^T_{t=1}  \mu_t 
\end{align*}
where we used the following bound on $\norm{\fcost^\star_i - \fcost^\star_{\mu_t}}_1$. Recall that $s$ is the vector as constructed in \Cref{remark:comparator}.%The idea of the bound is to upper bound the distance from the projection that it is not known analytically with the distance from an other $z \in\mathcal{X}_i^{\mu_t}$. And show that the latter distance is $\mathcal{O}(\mu_t)$. To construct such a point $z$ consider the collection of edges belonging in simple paths set $A$ and for each $e \in A$ sample a simple path containing $e$. After this procedure we collected at most $m$ simple paths denoted as $\{p_i\}^m_{i=1}$. Considering them in a vector form and considering the element-wise sum we obtain the vector $s = \frac{1}{m}\sum^m_{i=1} p_i$ that denotes how many times each edge has been traversed. Then, it holds that $s \geq 1/m$ and $s \in \mathcal{X}$ that implies $s \in \mathcal{X}^{1/m}$. At this point, we can define the point $z \triangleq (1 - m\mu_t)x_i^\star + m\mu_t s \geq \mu_t$ which is a convex combination since $\mu_t \leq \frac{1}{m} \quad \forall t$. Furthermore $z \in \mathcal{X}$ because it is given by a convex combination of $x^\star$ and $s$ that belongs in $\mathcal{X}$.
\begin{align*}
    \norm{x^\star_i - x^\star_{\mu_t}}_1 & = \norm{x^\star_i - (1 - m\mu_t)x^\star_i - m\mu_t s}_1 \\ & = \norm{m\mu_t (x^\star_i - s)}_1 \\&\leq m^2 \mu_t 
\end{align*}
\end{proof}

\thmcostconcentration*
\begin{proof}
We denote $\mathcal{F}^t$ a filtration adapted to the sigma algebra induced by the random variables $\bc{x_i^l}^t_{l=1}$ and let $\E_t$ denote expectation conditioned on $\mathcal{F}^t$. We recognize that $\innerprod{c^t - \hat{c}^t }{\fcost_i^{t} - \fcost^\star_{\mu_t}}$ is a bounded martingale difference sequence.
Indeed, \begin{equation*}
\E_t\bs{\innerprod{c^t - \hat{c}^t }{\fcost_i^{t} - \fcost^\star_{\mu_t}}} = \E_t\bs{\innerprod{\E_t\bs{\hat{\cost}^t} - \hat{\cost}^t }{\fcost_i^{t} - \fcost^\star_{\mu_t}}} = {\innerprod{\E_t\bs{\hat{\cost}^t} - \E_t\bs{\hat{\cost}^t} }{\fcost_i^{t} - \fcost^\star_{\mu_t}}} = 0
\end{equation*}
where the last equality holds because $x_i^t$ is $\mathcal{F}^t$-measurable. In addition, $\innerprod{c^t - \hat{c}^t }{\fcost_i^{t} - \fcost^\star_{\mu_t}}$ can be bounded as
\begin{equation*}
    \abs{\innerprod{c^t - \hat{c}^t}{\fcost_i^{t} - \fcost^\star_{\mu_t}}} \leq \norm{c^t - \hat{c}^t}_{\infty}\norm{\fcost_i^{t} - \fcost^\star_{\mu_t}}_1 \leq 2m c_{\max}(1 + 1/\mu_t)
\end{equation*}
where we used $\norm{\fcost_i^{t} - \fcost^\star_{\mu_t}}_1 \leq 2 m$ and $\norm{c^t - \hat{c}^t}_{\infty} \leq c_{\max}(1 + 1/\mu_t)$.
Therefore, by Azuma-Hoeffding, with probability $1-\delta$,
\begin{equation*}
     \sum^T_{t=1}\innerprod{ c^t - \hat{c}^t}{\fcost_i^{t} - \fcost^\star_{\mu_t}} \leq 2 m c_{\max}\sqrt{2\sum^T_{t=1} (1 + 1/\mu_t)^2\log(1/\delta)} \leq 2m  c_{\max} (1 + 1/\mu_T)\sqrt{2 T \log(1/\delta)}
\end{equation*}
\end{proof}
At this point, we have all the elements for the proof of \Cref{t:regret}
\subsection{Proof of \Cref{t:regret}}
\thmnoregret*
\begin{proof}
Combining the previous theorems we can obtained the following bounds
\begin{align*}
\sum_{t=1}^T\sum_{e \in p^t_i} c_e^t - \min_{p_i^\star \in \mathcal{P}_i}\sum_{e \in p^\star_i} c_e^t &\leq \sqrt{\frac{1}{2}c_{\max} \log \frac{m}{\delta_1} T} + \frac{m}{\gamma_T}  + \frac{ c^2_{\max} m}{2}\sum^T_{t=1}\frac{\gamma_t}{\mu^2_t} + m c_{\max}\sum^T_{t=1}  \mu_t \\&\phantom{=}+ 2m  c_{\max} (1 + 1/\mu_T)\sqrt{2 T \log(1/\delta)}
\end{align*}
Now, replacing $\gamma_t = t^{-3/5}$ and $\mu_t = \min\bc{1/m, t^{-1/5}}$, we obtain
\begin{align*}
\sum_{t=1}^T\sum_{e \in p^t_i} c_e^t - \min_{p_i^\star \in \mathcal{P}_i}\sum_{e \in p^\star_i} c_e^t &\leq \sqrt{\frac{1}{2}c_{\max} T \log \frac{m}{\delta_1} } + \frac{m}{T^{-3/5}}  + \frac{c^2_{\max} m}{2}\sum^T_{t=1}\frac{t^{2/5}}{t^{3/5}} + \frac{ c^2_{\max} m}{2}\sum^{m^{1/5}}_{t=1}\frac{m^2}{t^{3/5}} + m c_{\max}\sum^T_{t=1}  \frac{1}{t^{1/5}} \\&\phantom{=}+ m c_{\max}\sum^{m^{1/5}}_{t=1} \frac{1}{m}+ 2m  c_{\max} (1 + 1/\mu_T)\sqrt{2 T \log(1/\delta)} \\ &= \sqrt{\frac{1}{2}c_{\max} T \log \frac{m}{\delta_1} } + m^2T^{3/5}  + \frac{6 c_{\max} m}{8}T^{4/5} + \frac{ c_{\max} m}{2} m^{11/5} + m c_{\max}\frac{5}{4}T^{4/5}\\&\phantom{=}+ m^{1/5} c_{\max}+ 2m  c_{\max} (1 + T^{1/5})\sqrt{2 T \log(1/\delta)} 
\end{align*}
which of the order stated in the main text.
\end{proof}

\section{Proof for \Cref{sec:second_thm_sketch}}

\lemmasmoothness*
\begin{proof}
\if 0 
For a fixed entry of the vector $x$ denoted as $x_{\bar{j}e}$, we can compute the partial derivative as follows
\begin{equation*}
    \frac{\partial \Phi (x)}{\partial x_{\bar{j}e}} = \alpha_e \br{1 + \sum^N_{j \neq \bar{j}} x_{je}} + b_e
\end{equation*}
At this point we have that
\begin{align*}
    \norm{\nabla \Phi (x) - \nabla \Phi (\bar{x})}_2 & = \sqrt{\sum^N_{\bar{j}= 1}\sum_{e\in \mathcal{E}} \br{\frac{\partial \Phi (x)}{\partial x_{\bar{j}e}} - \frac{\partial \Phi (\bar{x})}{\partial x_{\bar{j}e}}}^2} \\ & = \sqrt{\sum^N_{\bar{j}= 1}\sum_{e\in \mathcal{E}} \alpha^2_e\br{\sum^N_{j \neq \bar{j}} x_{je} - \bar{x}^{j,e}}^2}
    \\ & \leq \sqrt{\sum^N_{\bar{j}= 1}\sum_{e\in \mathcal{E}} \alpha^2_e\br{\sum^N_{j \neq \bar{j}} \abs{x_{je} - \bar{x}^{j,e}}}^2}
    \\ & = \sqrt{\sum^N_{\bar{j}= 1}\sum_{e\in \mathcal{E}} \alpha^2_e\br{\sum^N_{j \neq \bar{j}} \sqrt{\br{x_{je} - \bar{x}^{j,e}}^2}}^2}
     \\ & \leq \sqrt{\sum^N_{\bar{j}= 1}\sum_{e\in \mathcal{E}} \alpha^2_e\br{\sum^N_{j = 1} \sqrt{\br{x_{je} - \bar{x}^{j,e}}^2}}^2}
     \\ & \leq \sqrt{\sum^N_{\bar{j}= 1}\sum_{e\in \mathcal{E}} \alpha^2_e\br{\sqrt{\sum^N_{j = 1} n \br{x_{je} - \bar{x}^{j,e}}^2}}^2}
     \\ & = \sqrt{\sum^N_{\bar{j}= 1}\sum_{e\in \mathcal{E}} \alpha^2_e\sum^N_{j = 1} n \br{x_{je} - \bar{x}^{j,e}}^2}
     \\ & \leq \sum^N_{\bar{j}= 1}\alpha_{\max} \sqrt{n}\sqrt{\sum_{e\in \mathcal{E}} \sum^N_{j = 1}\br{x_{je} - \bar{x}^{j,e}}^2}
     \\ & = \alpha_{\max} n^{3/2}\norm{x - \bar{x}}_2
\end{align*}
\fi
\input{smoothness_proof.tex}
\end{proof}
%\begin{lemma}
%\label{lemma:proof_leonardos}
%Let $y^{t+1} = \Pi_{\mathcal{X}^{\mu_t}}\br{x^t - \lambda \nabla_{x} \Phi(y^{t+1})}$ and let $\lambda$ be the inverse of the smoothness parameter of $\Phi$ that is $\lambda = \frac{1}{\alpha_{\max} n^{3/2}}$ and $\Phi_{\max} = \lambda E c^2_{\max} + n c_{\max}$, then it holds true that
%\begin{equation*}
%    \min_{t\in [T]} \E \norm{y^{t+1} - x^t}_2 \leq \sqrt{\frac{\lambda^2 \Phi_{\max}}{2 \gamma T} + \frac{\lambda \gamma E c^2_{\max}}{2 \mu}}
%\end{equation*}
%\end{lemma}
\subsection{Proof of \Cref{lemma:conversion}}
Before proving \Cref{lemma:conversion} we need an auxiliary lemma (\Cref{lemma:conversion_bis}) for which we need to introduce the notion of \emph{Caratheodory's decomposition image}.

\begin{definition}\label{d:image}
We define as \emph{Caratheodory's decomposition image of the set $\mathcal{X}_i^\mu$}, denoted as $\Delta(\mathcal{P}_i^\mu)$, the set all all probability distributions $\pi_i \in \mathcal{P}_i$ such that 
\[ \Pp{p_i\sim\pi_i}{e \in p_i}\geq \mu ~~~~~\text{for all}~ e \in E_i.\]
\end{definition}

Now, we can state and prove the auxiliary lemma which relates $\epsilon,\mu$ stationary point to  $4n^2mc_{\max}\epsilon$-Nash equilibrium for strategies profile in the Caratheodory's decomposition of the set $\mathcal{X}^\mu$.
\begin{lemma} Let $(\pi_i, \pi_{-i}) \in \Delta(\mathcal{P}_1)\times \ldots \times \Delta(\mathcal{P}_n)$ be a Caratheodory decomposition of $x \in \mathcal{X}_1\times \ldots \times \mathcal{X}_n$ and let $x$ be a $(\epsilon, \mu)$-stationary point according to \Cref{def:fractional_potential}. Then for each agent $i\in [n]$,
\begin{equation*}
c_i(\pi_i, \pi_{-i}) - \min_{\bar{\pi}_i \in \Delta(\mathcal{P}_i^{\mu})} c_i(\bar{\pi}_i, \pi_{-i})\\\leq 4 n^2 m c_{\max} \norm{x - \Pi_{\mathcal{X}^\mu} \left[x - \frac{1}{2 n^2  c_{\max} \sqrt{m}}  \nabla \Phi(x) \right]}\
\end{equation*}
where $\Delta(\mathcal{P}_i^{\mu}) \subset \Delta(\mathcal{P}_i)$ is the Caratheodory's decomposition image of the set $\mathcal{X}^\mu_i$.
\label{lemma:conversion_bis}
\end{lemma}
\begin{proof}
Let a probability distribution $\bar{\pi_i} \in \Delta(\mathcal{P}_i^\mu)$ and its marginalization $\bar{x}_i$ i.e. $\bar{x}_{ie} = \Pp{p_i\sim\bar{\pi}_i}{e \in p_i}$. Then Definition~\ref{d:image} implies that $\bar{x}_i \in \mathcal{X}_i^\mu$.
As a result, \Cref{lemma:potential_fractional} implies 
\[ c_i\br{\pi_i, \pi_{-i}} - c_i\br{\bar{\pi}_i, \pi_{-i}}  = \Phi(\bar{x}_i, x_{-i}) - \Phi(x_i, x_{-i})\]

Let $\lambda$ be the inverse of the smoothness parameter of $\Phi$ that is $\lambda = \frac{1}{2 n^2  c_{\max} \sqrt{m}}$. To simplify notation let 
$\epsilon:=\norm{\Pi_{\mathcal{X}^{\mu}}[x - \lambda \nabla \Phi(x)] - x}_2$ meaning that $x$ is trivially an $(\epsilon, \mu)$-stationary point. Then, by the result \citep[Proposition B.1]{Agarwal:2019} and \citep[Lemma 3]{Ghadimi:2016}, we have that 
\begin{equation*}
    \max_{\delta \in \Delta} -\delta^T \nabla_{x_i} g(x_i) \leq 2\frac{\epsilon}{\lambda} \delta_{\max} \quad \forall i \in [n]
\end{equation*}
with $\Delta \triangleq \{\delta \quad\text{such that }\quad x_i + \delta \in \mathcal{X}_i^{\mu}, \norm{\delta}\leq \delta_{\max}\}$ and $g(x_i):= \Phi(x_i,x_{-i})$ is equal to the potential function $\Phi(\cdot)$ when all the players (but the player $i$) keep their strategy profile fixed. Since $\norm{x_i - x'_i} \leq \sqrt{m}$ for any $x,y \in \mathcal{X}_i$ we get that $\delta_{\max} \leq \sqrt{m}$. Now by plugging the value of $\lambda$, setting $\delta = x_i - \bar{x}_i$ on the left hand side, we get that for any player $i \in [n]$:
\begin{equation*}
    ( x_i - \bar{x}_{i})^T \nabla_{x_i} g(x_i) \leq 4 m n^2 c_{\max} \epsilon
\end{equation*}

The function $g(x_i)$ is a linear function (see Definition~\ref{def:fractional_potential}). By linearity of $g(x_i)$ it holds that
\begin{align*}
    \Phi(x) - \Phi(\bar{x}_{i}, x_{-i}) &= g(x_i) - g(\bar{x}_{i}) \\ &= ( x_i - \bar{x}_{i})^T \nabla_{x_i} g(x_i) \\ & \leq  4 m n^2 c_{\max} \epsilon
\end{align*}
As a result for any $\bar{\pi}_i \in \Delta(\mathcal{P}^\mu_i)$, 
\[ c_i\br{\pi_i, \pi_{-i}} - c_i\br{\bar{\pi}_i, \pi_{-i}}  \leq 4 m n^2 c_{\max} \epsilon\]

\end{proof}
We conclude the section by presenting a slightly more general version of Lemma~\ref{lemma:conversion}
\begin{lemma}
\label{lemma:conversion1}
Let $\pi = (\pi_1,\ldots,\pi_n) \in \Delta(\mathcal{P}_1)\times \ldots \times \Delta(\mathcal{P}_n)$ and  $x = (x_1,\ldots,x_n) \in \mathcal{X}_1 \times \ldots \times \mathcal{X}_n$ such that for all resources $e \in E$,
\[x_{i,e} = \Pp{p_i \sim \pi_i}{e \in p_i}.\]
Then the following holds, 
\begin{equation*}
c_i(\pi_i, \pi_{-i}) - \min_{\bar{\pi}_i \in \Delta(\mathcal{P}_i^{\mu})} c_i(\bar{\pi}_i, \pi_{-i})\\\leq 4 n^2 m c_{\max}\cdot \norm{x - \Pi_{\mathcal{X}^\mu} \left[x - \frac{1}{2 n^2  c_{\max} \sqrt{m}}  \nabla \Phi(x) \right]}\ + 2m^2nc_{\max}\mu
\end{equation*}
\end{lemma}

%\lemmaconversion*
\begin{proof}
\Cref{lemma:conversion_bis} implies that
\begin{equation*}
c_i(\pi_i, \pi_{-i}) - \min_{\bar{\pi}_i \in \Delta(\mathcal{P}_i^{\mu})} c_i(\bar{\pi}_i, \pi_{-i})\\\leq 4 n^2 m c_{\max} \cdot \norm{x - \Pi_{\mathcal{X}^\mu} \left[x - \frac{1}{2 n^2  c_{\max} \sqrt{m}}  \nabla \Phi(x) \right]}\ 
\end{equation*}
As a result, we just need to bound
\[\min_{\bar{\pi}_i \in \Delta(\mathcal{P}_i^{\mu})} c_i(\bar{\pi}_i, \pi_{-i}) - \min_{\bar{\pi}_i \in \Delta(\mathcal{P}_i)} c_i(\bar{\pi}_i, \pi_{-i})\]
Let $\bar{\pi}^\star_i \in \argmin_{\bar{\pi}_i \in \Delta(\mathcal{P}_i)} c_i(\bar{\pi}_i, \pi_{-i})$ and $\bar{x}^\star_i$ its marginalization. Let also $\bar{x}^\star_{i, \mu} := (1 - m \mu) \bar{x}^\star_{i} + m \mu s $ where the vector $s$ is defined as in the proof of \Cref{thm:bounded_poly} and $\bar{\pi}^\star_{i, \mu}$ its corresponding Caratheodory decomposition.
\begin{align*}
    \min_{\bar{\pi}_i \in \Delta(\mathcal{P}_i^{\mu})} c_i(\bar{\pi}_i, \pi_{-i}) &- \min_{\bar{\pi}_i \in \Delta(\mathcal{P}_i)} c_i(\bar{\pi}_i, \pi_{-i}) = \min_{\bar{\pi}_i \in \Delta(\mathcal{P}^i_{\mu})} c_i(\bar{\pi}_i, \pi_{-i}) -  c_i(\bar{\pi}^\star_i, \pi_{-i}) \\ &\leq c_i(\bar{\pi}^\star_{i,\mu}, \pi_{-i}) -  c_i(\bar{\pi}^\star_i, \pi_{-i})  \\&= \Phi(\bar{x}^\star_{i,\mu}, x_{-i}) - \Phi(\bar{x}^\star_i, x_{-i}) \\& = \sum_{e \in E} \sum_{\mathcal{S}\subset [n], i \in\mathcal{S}} \br{\bar{x}^\star_{ie} - (1 - m\mu)\bar{x}^\star_{ie} - m\mu s} \prod_{j\in \mathcal{S}, j \neq i} x_{je} \prod_{j\notin \mathcal{S}} ( 1- x_{je}) \sum^{\abs{\mathcal{S}}}_{l=0} c_e(l) \\ & \phantom{=} + \sum_{e \in E} \sum_{\mathcal{S}\subset [n], i \notin\mathcal{S}} \br{- \bar{x}^\star_{ie} + (1 - m\mu)\bar{x}^\star_{ie} + m \mu s} \prod_{j\in \mathcal{S}} x_{je} \prod_{j\notin \mathcal{S}, j \neq i} ( 1- x_{je}) \sum^{\abs{\mathcal{S}}}_{l=0} c_e(l) \\ & \leq  \sum_{e \in E} \br{\bar{x}^\star_{ie} - (1 - m \mu)\bar{x}^\star_{ie} - m \mu s} n c_{\max}\underbrace{\sum_{\mathcal{S}\subset [n], i \in\mathcal{S}} \prod_{j\in \mathcal{S}, j \neq i} x_{je} \prod_{j\notin \mathcal{S}} ( 1- x_{je})}_{=1}  \\ & \phantom{=} + \sum_{e \in E}  \br{- \bar{x}^\star_{ie} + (1 - m \mu)\bar{x}^\star_{ie} + m \mu s} n c_{\max}\underbrace{\sum_{\mathcal{S}\subset [n], i \notin\mathcal{S}} \prod_{j\in \mathcal{S}} x_{je} \prod_{j\notin \mathcal{S}, j \neq i} ( 1- x_{je})}_{=1} \\ & = 2 m^2 n c_{\max} \mu
\end{align*}.
\end{proof}
\subsection{Proof of \Cref{thm:estimator_props}}
\thmestimatorprops*
\begin{proof}
In the first part of the proof we show that the projection operator is separable. That is, for a generic set $\mathcal{X}$ $\Pi_{\mathcal{X}}(z) = \bs{\Pi_{\mathcal{X}_1}[z_1], \dots , \Pi_{\mathcal{X}_n}[z_n]}^\trans$ for any $z \in \mathbb{R}^{nm}$ in the form $z = \bs{z_1, \dots, z_n}$ with $z_j \in \mathbb{R}^m$ for all $j \in [n]$. To prove this, we proceed as follows
\begin{equation*}
    \Pi_{\mathcal{X}}[z] = \argmin_{z^\prime \in \mathcal{X}} \norm{z - z^\prime}^2 = \argmin_{z^\prime \in \mathcal{X}} \sum^n_{i=1}\norm{z_i - z_i^\prime}^2 = \sum^n_{i=1} \argmin_{z_i^\prime \in \mathcal{X}_i} \norm{z_i - z_i^\prime}^2 = \bs{\Pi_{\mathcal{X}_1}[z_1], \dots, \Pi_{\mathcal{X}_n}[z_n]}
\end{equation*}

Let $(\pi_i^t , \pi_{-i}^t)$ denote a Caratheodory Decomposition for $x^t := (x_i^t , x_{-i}^t)$. Then,
\begin{align*}
\E\left[ \left[ \nabla_t\right]_{ie} \right]&= \Pr_{\pi_i^t}\left[\text{agent }i \text{ selects resource }e \text{ at round }t 
\right] \cdot \E\left[c_e^t / x_{ie}^t ~|~\text{agent }i \text{ selects resource }e \text{ at round }t\right]\\
&= x_{ie}^t \cdot \E\left[c_e^t / x_{ie}^t ~|~\text{agent }i \text{ selects resource }e \text{ at round }t\right]\\
&= \E\left[c_e^t ~|~\text{agent }i \text{ selects resource }e \text{ at round }t\right]\\
&= \sum_{\mathcal{S}_{-i} \subseteq{[n - 1]}} \bigg(\prod_{j\in\mathcal{S}_{-i}} \Pr_{\pi_j^t}\left[\text{agent }j \text{ selects }e \text{ at round }t \right] \\&\phantom{=}\cdot\prod_{j\notin\mathcal{S}_{-i}} \Pr_{\pi_j^t}\left[\text{agent }j \text{ does not select }e \text{ at round }t \right] c_e(\abs{\mathcal{S}_{-i}} + 1)\bigg)\\
&= \sum_{\mathcal{S}_{-i} \subseteq{[n - 1]}} \prod_{j\in\mathcal{S}_{-i}} x_{je}^t\prod_{j\notin\mathcal{S}_{-i}} (1-x_{je}^t) c_e\left(\abs{\mathcal{S}_{-i}} + 1\right)
\end{align*}

At the same time,
\begin{align*}
\frac{\partial \Phi(x)}{\partial x_{ie}} &= \sum_{\mathcal{S}_{-i}\subseteq [n-1]} \prod_{j \in \mathcal{S}_{-i}} x_{je} \prod_{j \notin \mathcal{S}_{-i}} (1 - x_{je} ) \sum^{\abs{\mathcal{S}_{-i}} + 1}_{\ell=0} c_e(\ell) - \sum_{\mathcal{S}_{-i}\subseteq [n-1]} \prod_{j \in \mathcal{S}_{-i}} x_{je} \prod_{j \notin \mathcal{S}_{-i}} (1 - x_{je} ) \sum^{\abs{\mathcal{S}_{-i}}}_{\ell=0} c_e(\ell)  \\
&= \sum_{\mathcal{S}_{-i} \subseteq{[n - 1]}} \prod_{j\in\mathcal{S}_{-i}} x_{je}^t\prod_{j\notin\mathcal{S}_{-i}} (1-x_{je}^t) c_e\left(\abs{\mathcal{S}_{-i}} + 1\right)\\
&=\E\left[ \left[ \nabla_t\right]_{ie} \right]
\end{align*}
The second part of the proof concerns bounding the norm of the stochastic gradients of the potential function. More precisely, we show that $\E\left[ \norm{\nabla_t}^2 \right]\leq \frac{n c^2_{\max} m}{\mu_t}$
\begin{align*}
    \E\left[ \norm{\nabla_t}^2 \right] &= \sum^{N}_{i=1}\sum_{e \in E_i}\Pr_{\pi_i^t}\left[\text{agent }i \text{ selects resource }e \text{ at round }t 
\right] \cdot \E\left[(c_e^t / x_{ie}^t)^2 ~|~\text{agent }i \text{ selects resource }e \text{ at round }t\right]\\
&= \sum^{n}_{i=1}\sum_{e \in E_i}x_{ie}^t \cdot \E\left[(c_e^t / x_{ie}^t)^2 ~|~\text{agent }i \text{ selects resource }e\right]\\
&= \sum^{n}_{i=1}\sum_{e \in E_i} \E\left[(c_e^t)^2 / x_{ie}^t ~|~\text{agent }i \text{ selects resource }e\right]\\
&\leq \sum^{n}_{i=1}\sum_{e \in E_i} c_{\max}/\mu_t\\
&=  \frac{\sum^n_{i=1} \abs{E_i} c_{\max}^2}{\mu_t} \leq  \frac{nm c_{\max}^2}{\mu_t}
\end{align*}
\end{proof}

\subsection{Proof of \Cref{lemma:G_bound}}
\statpoint*
\begin{proof}
We make use of the Moreau envelope function defined as follows,
\begin{equation*}
    \phi_{\lambda \mathcal{X}^{\mu_{t+1}}}(x) := \min_{y\in\mathcal{X}^{\mu_{t+1}}} \br{ \Phi(y) + \frac{1}{\lambda}\norm{x - y}^2}
\end{equation*}
Let also $y^{t+1} := \argmin_{y\in\mathcal{X}^{\mu_{t+1}}} \br{  \Phi(y) + \frac{1}{\lambda}\norm{x^t - y}^2}$. It holds that
\begin{align*}
    \phi_{\lambda \mathcal{X}^{\mu_{t+1}}}(x^{t+1}) &\leq \Phi(y^{t+1}) + \frac{1}{\lambda}\norm{x^{t+1} - y^{t+1}}^2 \\ & \leq \Phi(y^{t+1}) + \frac{1}{\lambda}\norm{x^{t} - \gamma_t \nabla_t - y^{t+1}}^2 \\ & =  \Phi(y^{t+1}) + \frac{1}{\lambda}\norm{x^{t} - y^{t+1}}^2 + \frac{\gamma_t^2}{\lambda} \norm{\nabla_t}^2 - \frac{2\gamma_t}{\lambda}(x^{t} - y^{t+1})^T\nabla_t \\ & = \phi_{\lambda\mathcal{X}^{\mu_{t+1}}}(x^{t})  + \frac{\gamma_t^2}{\lambda} \norm{\nabla_t}^2 - \frac{2 \gamma_t}{\lambda}(x^{t} - y^{t+1})^T\nabla_t .
\end{align*}

where the first inequality comes from the definition of Moreau envelope, the second inequality comes from projection property on convex sets and last by the definition $y^{t+1} = \argmin_{y\in\mathcal{X}^{\mu_{t+1}}} \br{  \Phi(y) + \frac{1}{\lambda}\norm{x^t - y}^2}$.

Then, taking total expectation on both sides and using the monotonicity property of expectation, we have
\begin{align*}
    \E \bs{\phi_{\lambda\mathcal{X}^{\mu_{t+1}}}(x^{t+1})} & \leq \E \bs{\phi_{\lambda\mathcal{X}^{\mu_{t+1}}}(x^{t})}  + \frac{\gamma_t^2}{\lambda} \E \bs{\norm{\nabla_t}^2} - \frac{2 \gamma_t}{\lambda}\E \bs{ (x^{t} - y^{t+1})^T\E\bs{ \nabla_t | x^t}}.
\end{align*}
At this point using the bound on the expected squared norm (see~\Cref{thm:estimator_props}) of the cost estimator and using that the cost estimator is unbiased we obtain
\begin{align*}
    \E \bs{\phi_{\lambda\mathcal{X}^{\mu_{t+1}}}(x^{t+1})} & \leq \E \bs{\phi_{\lambda\mathcal{X}^{\mu_{t+1}}}(x^{t})}  + \frac{\gamma_t^2}{\lambda} \frac{c^2_{\max} E}{\mu_t} - \frac{2\gamma_t}{\lambda}\E \bs{ (x^{t} - y^{t+1})^T \nabla \Phi(x^t)}
\end{align*}
At this point, using the fact that $\Phi(\cdot)$ is $\frac{1}{\lambda}$-smooth we get that
\begin{equation*}
    (y^{t+1}-x^{t})^T \nabla \Phi(x^t) \leq \Phi(y^{t+1}) - \Phi(x^{t}) + \frac{1}{2 \lambda} \norm{y^{t+1} - x^{t}}^2
\end{equation*}
which implies that
\begin{align*}
    \E \bs{\phi_{\lambda\mathcal{X}^{\mu_{t+1}}}(x^{t+1})} & \leq \E \bs{\phi_{\lambda\mathcal{X}^{\mu_{t+1}}}(x^{t})}  + \frac{\gamma_t^2}{\lambda} \frac{c^2_{\max} n m}{\mu_t} + \frac{2 \gamma_t}{\lambda}\E \bs{  \Phi(y^{t+1}) - \Phi(x^{t}) + \frac{1}{2 \lambda} \norm{y^{t+1} - x^{t}}^2}.
\end{align*}
Due to the fact that $\mathcal{X}^{\mu_{t}} \subseteq \mathcal{X}^{\mu_{t+1}}$ we get that $\phi_{\lambda\mathcal{X}^{\mu_{t+1}}}(x^{t}) \leq \phi_{\lambda\mathcal{X}^{\mu_{t}}}(x^{t})$ and thus
\begin{align*}
    \E \bs{\phi_{\lambda\mathcal{X}^{\mu_{t+1}}}(x^{t+1})} & \leq \E \bs{\phi_{\lambda\mathcal{X}^{\mu_{t}}}(x^{t})}  + \frac{\gamma_t^2}{\lambda} \frac{c^2_{\max} n m}{\mu_t} + \frac{2 \gamma_t}{\lambda}\E \bs{ \Phi(y^{t+1}) - \Phi(x^{t}) + \frac{1}{2 \lambda} \norm{y^{t+1} - x^{t}}^2}.
\end{align*}
By reordering the terms and summing from $t=1$ to $T$ we get that
\begin{align} 
-\frac{2}{\lambda T}\sum^T_{t=1}\gamma_t\E \bs{ \Phi(y^{t+1}) - \Phi(x^{t}) + \frac{1}{2 \lambda} \norm{y^{t+1} - x^{t}}^2} & \leq \frac{\E \bs{\phi_{\lambda\mathcal{X}^{\mu_1}}(x_{1})} - \E \bs{\phi_{\lambda\mathcal{X}^{\mu_T}}(x_{T})}}{T}  +  \frac{c^2_{\max} n m}{\lambda T} \sum^T_{t=1}\frac{\gamma_t^2}{\mu_t} \nonumber \label{eq:bound_1} \\ & \leq \frac{\Phi_{\max}}{T}  +  \frac{c^2_{\max} n m}{\lambda T} \sum^T_{t=1}\frac{\gamma_t^2}{\mu_t}.
\end{align}
since $\phi_{\lambda \mathcal{X}^{\mu_1}}(x_1) \leq \Phi(x_1) \leq n m c_{\max}$.

\if 0
\textcolor{red}{
we used that $\E \phi_{\lambda\mathcal{X}_{\mu_0}}(x_{0}) - \E \phi_{\lambda\mathcal{X}^{\mu_t}}(x_{T}) \leq 2\Phi_{\max} + \frac{m}{\lambda} = \mathcal{O}\br{\phi_{\max}}$ because
\begin{align*}
\E \phi_{\lambda\mathcal{X}_{\mu_0}}(x_{0}) - \E \phi_{\lambda\mathcal{X}^{\mu_t}}(x_{T}) &\leq \min_{y\in\mathcal{X}_{\mu_0}} \bs{\Phi(y) + \frac{1}{\lambda}\norm{x_0 - y}^2} - \min_{y\in\mathcal{X}^{\mu_t}} \E\bs{\Phi(y) + \frac{1}{\lambda}\norm{x^T - y}^2} \\&= - \max_{y\in\mathcal{X}_{\mu_0}} \bs{-\Phi(y) - \frac{1}{\lambda}\norm{x_0 - y}^2} + \max_{y\in\mathcal{X}^{\mu_t}} \E\bs{-\Phi(y) - \frac{1}{\lambda}\norm{x^T - y}^2} \\&\leq - \max_{y\in\mathcal{X}_{\mu_0}} \bs{-\Phi(y) - \frac{1}{\lambda}\norm{x_0 - y}^2} + \Phi_{\max} \\&\leq \frac{1}{\lambda}\norm{x_0 - y^1}^2 + \Phi_{\max} + \Phi(y^1) \\&\leq \frac{m}{\lambda} + 2\Phi_{\max} = \mathcal{O}\br{\phi_{\max}}
\end{align*}
where the last equality is proven replacing $\lambda$.}
\fi
\if 0
Then, defining as $t^\star$ the time index minimizing the left hand side an using that $\E \phi_\lambda(x_{0}) - \E \phi_\lambda(x_{T}) \leq \Phi_{\max}$, we obtain
\begin{align} 
-\frac{2 \gamma}{\lambda}\E \bs{ \Phi(x^{t^\star}) - \Phi(y_{t^\star+1}) + \frac{1}{2 \lambda} \norm{y_{t^\star+1} - x^{t^\star}}^2} & \leq \frac{\Phi_{\max}}{T}  + \frac{\gamma^2}{\lambda} \frac{c^2_{\max} n m}{\mu} .

\end{align}
%hi luca, are you there? Yes Hi YannisYes Can we talk in messenger?
\fi
Since $\Phi(x)$ is $\frac{1}{\lambda}$-smooth the function $H(x) = \Phi(x) + \frac{1}{\lambda} \norm{x - x^{t}}^2$ is convex. Now the definition $y^{t+1} = \argmin_{y \in \mathcal{X}^{\mu_{t+1}}} H(y)$ that implies $\innerprod{\nabla H (y^{t+1})}{x - y^{t+1}} \geq 0$ for all $x\in\mathcal{X}^{\mu_{t+1}}$. The latter holds also for all $x\in\mathcal{X}^{\mu_t}$ since $ \mathcal{X}^{\mu_t} \subseteq \mathcal{X}^{\mu_{t+1}}$. At this point we get that,
\begin{align*}
    \Phi(x^{t}) - \Phi(y^{t+1})   - \frac{1}{2 \lambda} \norm{y^{t+1} - x^{t}}^2 & =    H(x^{t}) - H (y^{t+1})  + \frac{1}{2 \lambda} \norm{y^{t+1} - x^{t}}^2 \\ & \geq \nabla H (y^{t+1})^T(  x^{t} - y^{t+1}) + \frac{1}{2\lambda} \norm{y^{t+1} - x^{t}}^2 \\ & \geq \frac{1}{2\lambda} \norm{y^{t+1} - x^{t}}^2
\end{align*}
Therefore, plugging into \eqref{eq:bound_1}, it holds that
\begin{equation*}
    \frac{1 }{T \lambda^2} \sum^T_{t=1}\gamma_t\E \norm{y^{t+1} - x^{t}}^2  \leq \frac{n m c_{\max}}{T}  + \frac{c^2_{\max} n m}{\lambda T} \sum^T_{t=1}\frac{\gamma_t^2}{\mu_t}.
\end{equation*}
Then, using the lower bound $\frac{1}{T \lambda^2} \sum^T_{t=1}\gamma_t\E \norm{y^{t+1} - x^{t}}^2 \geq \frac{\gamma_T}{T \lambda^2} \sum^T_{t=1}\E \norm{y^{t+1} - x^{t}}^2$ on the left hand side, using Jensen's inequality and taking square root on both sides we obtain
\begin{equation}
    \frac{1}{T} \sum^T_{t=1}\E \norm{y^{t+1} - x^{t}}  \leq \sqrt{\frac{\lambda^2 n m c_{\max}}{T \gamma_T}  + \frac{\lambda c^2_{\max} n m}{ T \gamma_T} \sum^T_{t=1}\frac{\gamma_t^2}{\mu_t}}. \label{eq:avg_bound}
\end{equation}
%\begin{equation*}
%    \E \norm{y^{t+1} - x^{t}} \leq \sqrt{\frac{\lambda^2 \Phi_{\max}}{2 \gamma T} + \frac{\lambda \gamma E c^2_{\max}}{2 \mu}},
%\end{equation*}
Before concluding the proof we need to bound the difference between the elements in the sequence $y^t$ and the iterates of proximal point projected always on the final set $\mathcal{X}^{\mu_t}$. To this end, we introduce the sequence $\tilde{y}^{t+1} = \Pi_{\mathcal{X}^{\mu_{T}}}\bs{x^t - \frac{\lambda}{2}\nabla\Phi(\tilde{y}^{t+1})}$ and we notice that
\begin{equation*}
    \norm{y^{t+1} - \tilde{y}^{t+1}} = \norm{\Pi_{\mathcal{X}_{\mu_{t+1}}}\bs{x^t - \frac{\lambda}{2}\nabla\Phi(y^{t+1})}- \Pi_{\mathcal{X}_{\mu_{T}}}\bs{x^t - \frac{\lambda}{2}\nabla\Phi(\tilde{y}^{t+1})}}
\end{equation*}
At this point by defining $w^t := x^t - \frac{\lambda}{2}\nabla\Phi(y^{t+1})$ and $\tilde{w}^t := x^t - \frac{\lambda}{2}\nabla\Phi(\tilde{y}^{t+1})$ we get that
\begin{align*}
\norm{y^{t+1} - \tilde{y}^{t+1}} &\leq\norm{\Pi_{\mathcal{X}}[w^t] - \Pi_{\mathcal{X}^{\mu_{t+1}}}[w^t]} + \norm{\Pi_{\mathcal{X}}[\tilde{w}^t] - \Pi_{\mathcal{X}^{\mu_{T}}}[\tilde{w}^t]} + \norm{\Pi_{\mathcal{X}}[w^t] - \Pi_{\mathcal{X}}[\tilde{w}^t]} \\& \leq \sqrt{nm^3} \mu_{t+1} + \sqrt{nm^3} \mu_{T} + \norm{w^t - \tilde{w}^t} \\& = \sqrt{nm^3} \mu_{t+1} + \sqrt{nm^3} \mu_{T} + \frac{\lambda}{2}\norm{\nabla\Phi(y^{t+1}) - \nabla\Phi(\tilde{y}^{t+1})} \\& \leq \sqrt{nm^3} \mu_{t+1} + \sqrt{nm^3} \mu_{T} + \frac{1}{2}\norm{y^{t+1} - \tilde{y}^{t+1}}
\end{align*}
where in the first inequality, we used the bound distance between a point in $\mathcal{X}$ and its projection in $\mathcal{X}^{\mu}$ used in the proof of \Cref{thm:bounded_poly} and in the last inequality we used the Lipschitz continuity of the gradients of the potential function.
The above estimation implies \begin{equation}\norm{y^{t+1} - \tilde{y}^{t+1}}  \leq 2 \sqrt{nm^3} (\mu_{t+1} + \mu_{T})\label{eq:y_tilde_y_bound}.\end{equation} Moreover, by a simple application of triangular inequality and the bounds in \Cref{eq:avg_bound} and in \Cref{eq:y_tilde_y_bound}, we obtain 
\begin{align*}
\frac{1}{T} \sum^T_{t=1}\E\norm{x^t- \tilde{y}^{t+1}} &\leq \frac{1}{T} \sum^T_{t=1}\E\norm{y^{t+1} - \tilde{y}^{t+1}} + \frac{1}{T} \sum^T_{t=1}\E \norm{y^{t+1} - x^{t}}  \\&\leq \frac{1}{T} \sum^T_{t=1} 2 \sqrt{nm^3} (\mu_{t+1} + \mu_{T}) + \sqrt{\frac{\lambda^2 n m c_{\max}}{2 T \gamma_T}  + \frac{\lambda c^2_{\max} n m}{ 2 T \gamma_T} \sum^T_{t=1}\frac{\gamma_t^2}{\mu_t}} \\&\leq \frac{4 \sqrt{nm^3} }{T} \sum^T_{t=1} \mu_{t} + \sqrt{\frac{\lambda^2n m c_{\max}}{2 T \gamma_T}  + \frac{\lambda c^2_{\max} n m}{ 2 T \gamma_T} \sum^T_{t=1}\frac{\gamma_t^2}{\mu_t}}
\end{align*}
Finally, we conclude the proof with the following steps
\begin{align*}
     \frac{1}{T}\sum^T_{t=1}\E \norm{G(x^t)}_2 &= \frac{1}{T}\sum^T_{t=1}\E \norm{\Pi_{\mathcal{X}^{\mu_t}}\bs{x^{t} - \frac{\lambda}{2} \nabla \Phi(x^{t})} - x^{t}}_2  \\&\leq \frac{1}{T}\sum^T_{t=1}\E \norm{\Pi_{\mathcal{X}^{\mu_t}}\bs{x^{t} - \frac{\lambda}{2} \nabla \Phi(x^{t})} - \tilde{y}^{t+1}}_2 + \frac{1}{T}\sum^T_{t=1}\E \norm{\tilde{y}^{t+1} - x^{t}}_2 \\ & \leq \frac{1}{T}\sum^T_{t=1}\E \norm{\Pi_{\mathcal{X}^{\mu_t}}\bs{x^{t} - \frac{\lambda}{2} \nabla \Phi(x^{t})} - \Pi_{\mathcal{X}^{\mu_t}}\bs{x^{t} - \frac{\lambda}{2} \nabla \Phi(\tilde{y}^{t+1})}}_2\\ &\phantom{\leq} + \sqrt{\frac{\lambda^2 n m c_{\max}}{2 T \gamma_T}  + \frac{\lambda c^2_{\max} n m}{ 2 T \gamma_T} \sum^T_{t=1}\frac{\gamma_t^2}{\mu_t}} + \frac{4 \sqrt{nm^3} }{T} \sum^T_{t=1} \mu_{t} \\ & \leq \frac{\lambda }{2 T}\sum^T_{t=1}\E \norm{ \nabla \Phi(x^{t}) - \nabla \Phi(\tilde{y}^{t+1}) }_2 + \sqrt{\frac{\lambda^2 n m c_{\max}}{2 T \gamma_T}  + \frac{\lambda c^2_{\max} n m}{ 2 T \gamma_T} \sum^T_{t=1}\frac{\gamma_t^2}{\mu_t}} + \frac{4 \sqrt{nm^3} }{T} \sum^T_{t=1} \mu_{t} \\ & \leq \frac{\lambda}{2 T}\frac{1}{\lambda}\sum^T_{t=1}\E\norm{x^{t} - \tilde{y}^{t+1}}_2+ \sqrt{\frac{\lambda^2 n m c_{\max}}{2 T \gamma_T}  + \frac{\lambda c^2_{\max} n m}{ 2 T \gamma_T} \sum^T_{t=1}\frac{\gamma_t^2}{\mu_t}} + \frac{4 \sqrt{nm^3} }{T} \sum^T_{t=1} \mu_{t} \\ & \leq
     2\sqrt{\frac{\lambda^2 n m c_{\max}}{2 T \gamma_T}  + \frac{\lambda c^2_{\max} n m}{ 2 T \gamma_T} \sum^T_{t=1}\frac{\gamma_t^2}{\mu_t}} + \frac{8 \sqrt{nm^3} }{T} \sum^T_{t=1} \mu_{t} 
\end{align*}
that concludes the proof.
\end{proof}

\subsection{Proof of \Cref{thm:convergence_to_nash}}
\thmconvergencetoNash*
\begin{proof}
Let $x^t$ denote the marginalization of $\pi^t$ then by applying Lemma~\ref{lemma:conversion1} for $\mu:= \mu_T$ we get that
\begin{align*}
\max_{i \in [n]}\left[c_i(\pi^{t}_i, \pi^{t}_{-i}) - \min_{\pi_i \in \Delta(\mathcal{P}_i)} c_i(\pi_i, \pi_{-i}^{t})\right]\leq & 4 n^2mc_{\max} \norm{G(x^t)} + 2m^2nc_{\max}\mu_T
\end{align*} 

As a result,
\begin{align*}
\frac{1}{T}\E \left[\sum_{t=1}^T\max_{i \in [n]}\left[c_i(\pi^{t}_i, \pi^{t}_{-i}) - \min_{\pi_i \in \Delta(\mathcal{P}_i)} c_i(\pi_i, \pi_{-i}^{t})\right]\right]\leq & 4 n^2mc_{\max} \E\left[ \frac{1}{T} \sum_{t=1}^T \norm{G(x^t)}\right] + 2m^2nc_{\max}\mu_T
\end{align*} 
where $G(x) = \Pi_{\mathcal{X}^{\mu_T}}\bs{x - \lambda \nabla \Phi(x)} - x$. Then by Theorem~\ref{lemma:G_bound} and the fact that $\lambda = (2n^2c_{\max}\sqrt{m})^{-1}$ we get

\begin{align*}
\frac{1}{T}\E \left[\sum_{t=1}^T\max_{i \in [n]}\left[c_i(\pi^{t}_i, \pi^{t}_{-i}) - \min_{\pi_i \in \Delta(\mathcal{P}_i)} c_i(\pi_i, \pi_{-i}^{t})\right]\right]\leq & 4 \sqrt{m} \sqrt{\frac{n m c_{\max}}{2 T \gamma_T}  + \frac{ c^2_{\max} nm}{ 2 T \gamma_T \lambda} \sum^T_{t=1}\frac{\gamma_t^2}{\mu_t}} + \frac{16 \sqrt{n} m^2}{T \lambda} \sum^T_{t=1} \mu_{t} \\ &  + 2 m^2 n c_{\max} \mu_T 
\end{align*} 
To simplify notation let
\[(\mathrm{A}):= \frac{1}{T}\E \left[\sum_{t=1}^T\max_{i \in [n]}\left[c_i(\pi^{t}_i, \pi^{t}_{-i}) - \min_{\pi_i \in \Delta(\mathcal{P}_i)} c_i(\pi_i, \pi_{-i}^{t})\right]\right]\]

At this point by choosing the sequence $\gamma_t = C_\gamma t^{-3/5}$ and $\mu_t = C_\mu\min\bc{1/m, t^{-1/5}}$ we have that
\begin{align*}
(\mathrm{A}) &\leq 4 \sqrt{m} \sqrt{\frac{n m c_{\max}}{2 T^{2/5}C_\gamma}  + \frac{ c^2_{\max} nm C_\gamma}{ 2 T^{2/5} \lambda C_\mu} \sum^{m^{1/5}}_{t=1}\frac{1}{mt^{6/5}} + \frac{ c^2_{\max} nm C_\gamma}{ 2 T^{2/5} \lambda C_\mu} \sum^T_{t=1}\frac{t^{1/5}}{t^{6/5}}} \\&\phantom{=}+ \frac{20 \sqrt{n}m^2 C_\mu}{\lambda T^{1/5}}  + \frac{2 m^2 n c_{\max} C_\mu}{T^{1/5}} \\
    &\leq 4 \sqrt{m} \sqrt{\frac{n m c_{\max}}{2 T^{2/5}C_\gamma}  + \frac{c^2_{\max} nmC_\gamma}{ 2 T^{2/5}\lambda C_\mu} (\log T + 1)} + \frac{20 \sqrt{n}m^2 C_\mu}{\lambda T^{1/5}} + \frac{2 m^2 n c_{\max} C_\mu}{T^{1/5}} \\
    &= \br{4 \sqrt{m} \sqrt{\frac{n m c_{\max}}{2 C_\gamma} + \frac{ c^2_{\max} nm C_\gamma}{ 2 \lambda C_\mu} (\log T + 1)} + \frac{20 \sqrt{n}m^2 C_\mu}{\lambda} + 2 m^2 n c_{\max} C_\mu} T^{-1/5}
\end{align*}
By replacing the values of $\lambda = (2 n^2 c_{\max} \sqrt{m})^{-1}$, we obtain
\begin{align*}
(\mathrm{A}) &\leq \br{4 \sqrt{m} \sqrt{\frac{ mn c_{\max}}{2 C_\gamma } +  \frac{c^3_{\max} n^3m^{3/2} C_\gamma(\log T + 1)}{C_\mu} } + 80 m^{5/2} n^{5/2} c_{\max}C_\mu } T^{-1/5}
\end{align*}
By neglecting non-dominant terms and choosing $C_{\gamma} = (m^{4/5} n^{8/5}c_{\max})^{-1}$ and $C_{\mu} = (n^{6/5} m^{11/10})^{-1}$, we obtain:
\begin{align*}
(\mathrm{A}) &\leq \br{4 \sqrt{m} \sqrt{\frac{ c^2_{\max} n^{13/5}m^{9/5}}{2 } +  c^2_{\max} n^{13/5}m^{9/5}(\log T + 1)} + 80 m^{14/10} n^{13/10} c_{\max}} T^{-1/5} \\&\leq  \frac{88 m^{14/10} n^{13/10} c_{\max} (\log T + 1) }{T^{1/5} }
\end{align*}
%The dependence on problem dependent parameter can be improved choosing the sequences 
The latter implies that if $T = \mathcal{O}\left( m^{7}n^{6.5}/\epsilon^5 \right)$ then $(\mathrm{A}) \leq \epsilon$.

By the choosing $C_{\gamma}=C_\mu=1$ we obtain
\[(\mathrm{A}) \leq \mathcal{O}\br{\frac{\sqrt{m^5 n^3 c^3_{\max}}(\log T + 1) }{T^{1/5} }}\]
which implies that if $T = \mathcal{O}\left( m^{12.5}n^{7.5}/\epsilon^5 \right)$ then $(\mathrm{A}) \leq \epsilon$.
\end{proof}

\Markov*
\begin{proof}
Let the random variable $E_t:= \max_{i \in [n]}\left[c_i(\pi^{t}_i, \pi^{t}_{-i}) - \min_{\pi_i \in \Delta(\mathcal{P}_i)} c_i(\pi_i, \pi_{-i}^{t})\right]$. Consider the random variable $E$ taking the value of the random variable $E_t$ with $t$ being sampled uniformly at random. Notice that 
\[\E [E] = \frac{1}{T}\E \left[\sum_{t=1}^T\max_{i \in [n]}\left[c_i(\pi^{t}_i, \pi^{t}_{-i}) - \min_{\pi_i \in \Delta(\mathcal{P}_i)} c_i(\pi_i, \pi_{-i}^{t})\right]\right] \leq \epsilon.\]
which by Markov inequality implies that with probability $\geq 1-\delta$, $E \leq \epsilon/\delta$. Thus with probability $\geq 1 -\delta$, $\max_{i \in [n]}\left[c_i(\pi^{t}_i, \pi^{t}_{-i}) - \min_{\pi_i \in \Delta(\mathcal{P}_i)} c_i(\pi_i, \pi_{-i}^{t})\right] \leq \epsilon/\delta$ meaning that $\pi^t$ is an $\epsilon/\delta$-Mixed NE. This establishes the second item of Corollary~\ref{c:markov}. Now consider the set of time steps $\mathcal{B}:= \{t \in \{1,t\}:~ E_t > \epsilon/\delta^2\}$. With probability $1-\delta$, $\sum_{t=1}^T E_t \leq \frac{\epsilon T}{\delta}$ we directly get that we probability $1-\delta$, $|\mathcal{B}| \leq \delta T$. As a result, with probability $\geq 1-\delta$, $(1-\delta)$ fraction of the profiles $\pi^1,\ldots.\pi^T$ are $\epsilon/\delta^2$-Mixed NE.
\end{proof}

\section{Auxiliary Lemmas}
\if 0
\begin{lemma}
\label{lemma:feasibility} $e \in \mathcal{A}$ if and only if the set  $\bc{ x\in\mathcal{X} : x_e = 1 }$ is not empty.
\end{lemma}
\begin{proof}
$\br{\Rightarrow}$ If $e \in \mathcal{A}$, this means that $e \in \hat{P}$ with $\hat{P}$ being a simple path that is feasible in $\mathcal{X}$ with $\hat{P}_e = 1$.

$\br{\Leftarrow}$ Let assume that $x\in\bc{\mathcal{X} : x_e = 1}$. Then, by Caratheodory's theorem, we can write $x=\sum^{E+1}_{j=1} \lambda^j \hat{P}^j$ with $\bc{\hat{P}^j}^{E+1}_{j=1}$ being simple paths.
At this point, it must exists at least one index $j^\star$ such that $\hat{P}_e^{j^\star} = 1$. This means that there exists a simple path containing the edge $e$. Then, it follows that $e \in \mathcal{A}$.
\end{proof}
\fi
\begin{lemma}
\label{lemma:path_existence}
Let the set $A$ be defined as in \Cref{alg:charateodory}, then $\forall e \in A$, given $x\in\mathcal{X}^\mu$ such that $x_e > 0$,  there exists a simple path $\hat{p}$ such that \begin{itemize}
    \item (i) $e \in \hat{p}$, 
    \item (ii) $x_e > 0 \quad \forall e \in \hat{p}$.
\end{itemize}
Therefore, Step 6 in \Cref{alg:charateodory} can always be implemented.
\end{lemma}
\begin{proof}
%\looseness=-1
By the Caratheodory's theorem, there exists a collection of simple paths $\{ \hat{p}_1, \dots,  \hat{p}_{m+1}\}$ and scalars $\lambda_1, \dots, \lambda_{m+1}$ such that $\lambda_i \geq 0 \quad \forall i \in \{1, m+1\}$ and $\sum^{m+1}_{i=1} \lambda_i = 1$ that allows to write $x = \sum^{m+1}_{j=1} \lambda_j \hat{p}_j$.
At this point, assume by contradiction that $\hat{p}_{je} = 0$ for all $j \in \{1, m+1\}$. This implies that $x_e = 0$ which is a clear contradiction.
That means that there exist $j^\star \in \{1, m +1\}$ such that $e \in \hat{p}_{j^\star}$ proving part (i). In addition it must be true that the weight in the convex combination is positive, i.e. $\lambda_{j^\star} > 0$. This implies that $x_e \geq \lambda_{j^\star}\hat{p}_{j^\star e}$. Therefore, for the edges $e$ s.t. $\hat{p}_{j^\star e} = 1$, it holds that $x_e \geq \lambda_{j^\star} > 0$.
\end{proof}

%\section{Proof of Lemma~\ref{lemma:xmu_not_empty}}

\begin{lemma}
\label{lemma:potential_fractional}
Let $x,\bar{x} \in \mathcal{X}$ and $\pi,\bar{\pi}$ be their respective Caratheodory decompositions. Then,
\begin{equation*}
c_i\br{\bar{\pi}_i, \pi_{-i}} - c_i\br{\pi_i, \pi_{-i}}  = \Phi(\bar{x}_i, x_{-i}) - \Phi(x_i, x_{-i})
\end{equation*}
\end{lemma}

\begin{proof}
We start by manipulating the cost difference 
\begin{align*}
c_i\br{\bar{\pi}_i, \pi_{-i}} - c_i\br{\pi_i, \pi_{-i}} &= \sum_{p_i\in\mathcal{P}_i} \sum_{p_{-i}\in\mathcal{P}_{-i}} \Pp{\pi_{-i}}{p_{-i}} \br{\Pp{\bar{\pi}_i}{p_i}C_i(p_i,p_{-i}) - \Pp{\pi_i}{p_i}C_i(p_i,p_{-i})} \\&= \sum_{p_i\in\mathcal{P}_i}\sum_{p'_i\in\mathcal{P}_i} \sum_{p_{-i}\in\mathcal{P}_{-i}} \Pp{\pi_{-i}}{p_{-i}}\Pp{\bar{\pi}_i}{p'_i}\Pp{\pi_i}{p_i} \br{C_i(p'_i,p_{-i}) - C_i(p_i,p_{-i})}
 \\&= \sum_{p_i\in\mathcal{P}_i}\sum_{p'_i\in\mathcal{P}_i} \sum_{p_{-i}\in\mathcal{P}_{-i}} \Pp{\pi_{-i}}{p_{-i}}\Pp{\bar{\pi}_i}{p'_i}\Pp{\pi_i}{p_i} \br{\Phi(p'_i,p_{-i}) - \Phi(p_i,p_{-i})}
 \\&= \sum_{e \in E}\sum_{p_i\in\mathcal{P}_i}\sum_{p'_i\in\mathcal{P}_i} \sum_{p_{-i}\in\mathcal{P}_{-i}} \Pp{\pi_{-i}}{p_{-i}}\Pp{\bar{\pi}_i}{p'_i}\Pp{\pi_i}{p_i} \br{\sum_{i=0}^{l_e(p'_i,p_{-i})} c_e(i) - \sum_{i=0}^{l_e(p_i,p_{-i})}c_e(i)} \\&= \sum_{e \in E}\sum_{p'_i\in\mathcal{P}_i} \sum_{p_{-i}\in\mathcal{P}_{-i}} \Pp{\pi_{-i}}{p_{-i}}\Pp{\bar{\pi}_i}{p'_i}\sum_{i=0}^{l_e(p'_i,p_{-i})} c_e(i) \\&\phantom{=}- \sum_{e \in E}\sum_{p_i\in\mathcal{P}_i} \sum_{p_{-i}\in\mathcal{P}_{-i}} \Pp{\pi_{-i}}{p_{-i}}\Pp{\pi_i}{p_i} \sum_{i=0}^{l_e(p_i,p_{-i})}c_e(i)
\end{align*}
At this point we can consider the two terms of the last expression can be written as the potential function in \Cref{def:fractional_potential}.
\begin{align*}
\sum_{e \in E}\sum_{p_i\in\mathcal{P}_i} \sum_{p_{-i}\in\mathcal{P}_{-i}} \Pp{\pi_{-i}}{p_{-i}}\Pp{\bar{\pi}_i}{p_i}\sum_{i=0}^{l_e(p_i,p_{-i})} c_e(i) &= \sum^N_{s=1}\sum_{e \in E}\sum_{p_i\in\mathcal{P}_i} \sum_{p_{-i}\in\mathcal{P}_{-i}} \Pp{\pi_{-i}}{p_{-i}}\Pp{\bar{\pi}_i}{p_i}\mathds{1}\bc{l_e(p_i,p_{-i}) = s}\sum_{i=0}^{s} c_e(i) \\&= \sum_{e \in E} \sum^N_{s=1} \Pp{\pi_i, \bar{\pi}_i}{\text{Exactly $s$ agents select edge $e$}}\sum_{i=0}^{s} c_e(i) \\&= \sum_{e \in E} \sum^N_{s=1} \sum_{\mathcal{S} \subset \binom{[n]}{s}} \prod_{j\in\mathcal{S}} x_{je}\prod_{j\notin\mathcal{S}} (1-x_{je})\sum_{i=0}^{s} c_e(i)
\\&= \sum_{e \in E} \sum_{\mathcal{S} \subset{[n]}} \prod_{j\in\mathcal{S}} x_{je}\prod_{j\notin\mathcal{S}} (1-x_{je})\sum_{i=0}^{\abs{\mathcal{S}}} c_e(i)
\end{align*}
\end{proof}
\section{On the difference with $\epsilon$-greedy exploration}

We present two examples to highlight the difference between $\epsilon$-greedy and Exploration with Bounded-Away Polytopes for the special case of simplex. 
Exploration with Bounded-Away Polytopes takes a mixed strategy $x\in \Delta_n$ and transforms it to the strategy
$$ x' := \Pi_{\Delta_n^\mu}( x) $$
where $\Delta_n^\mu = \Delta_n \cap \{x: x_i \geq \mu\}$. On the other hand $\epsilon$-greedy exploration transforms $x$ to $x'$ as follows,
$$ x' := (1-\epsilon) x + \frac{\epsilon}{n} (1,\ldots,1) $$
These are two different transformation that do not coincide. We will provide two specific examples : one example for $x \in  \Delta_n^\mu$ and one for $x \notin  \Delta_n^\mu$.

\textbf{Example for $x \in  \Delta_n^\mu$ :} Consider $\epsilon = \mu$ and $x=(2/3,1/3)$. In this case $\epsilon$-greedy exploration selects the strategy $x'= ( (1-\mu)2/3 + \mu/2 , (1-\mu)1/3 + \mu/2 )$ while Exploration with Bounded-Away Polytopes selects the strategy $x'=(2/3,1/3)$ because $x \in  \Delta_n^\mu$ implies that $x = x'$.

\textbf{Example for $x \notin  \Delta_n^\mu$ :} Consider $\epsilon = \mu$ and $x=(8/10,2/10, 0)$. In this case $\epsilon$-greedy exploration selects the strategy $x'= ( (1-\mu)8/10 + \mu/3 , (1-\mu)2/10 + \mu/3, \mu/3)$.
For Exploration with Bounded-Away Polytopes with $\mu \leq \frac{0.4}{3}$ we can use  the KKT conditions of the problem to derive that $x' = (x_1',x_2',x_3')$ must satisfy the following system
\begin{equation}
\begin{cases}
& 2(x^\prime_1 - 0.8) + \lambda = 0
\\& 2(x^\prime_2 - 0.2) + \lambda = 0
\\& 2(x^\prime_3 - 0.0) + \lambda \geq 0
\\& - x_1^\prime + \mu < 0
\\& - x_2^\prime + \mu < 0
\\& - x_3^\prime + \mu = 0
\end{cases} 
\end{equation}

which admits the solution $x' = (0.8 - \mu/2, 0.2 - \mu/2, \mu)$ which is a different transformation than the one obtained with $\epsilon$-greedy.
To see this, take for example $\epsilon=\mu=0.12$. $\epsilon$ greedy exploration gives $(( (1-\mu)8/10 + \mu/3 , (1-\mu)2/10 + \mu/3, \mu/3)) = (0.744, 0.216, 0.04)$ while Exploration with Bounded-Away Polytopes gives $(0.8 - \mu/2, 0.2 - \mu/2, \mu) = (0.74, 0.14, 0.12)$.
\clearpage
\newpage

\end{document}

%% file: math_commands.tex
\usepackage{amsmath,amsfonts,bm}
\usepackage{dsfont}
\newcommand{\bc}[1]{\left\{{#1}\right\}}
\newcommand{\br}[1]{\left({#1}\right)}
\newcommand{\bs}[1]{\left[{#1}\right]}
\newcommand{\abs}[1]{\left| {#1} \right|}

\newcommand{\norm}[1]{\left\| {#1} \right\|}
\def\eqref#1{(\ref{#1})}
\def\1{\bm{1}}

\DeclareMathAlphabet{\mathsfit}{\encodingdefault}{\sfdefault}{m}{sl}
\SetMathAlphabet{\mathsfit}{bold}{\encodingdefault}{\sfdefault}{bx}{n}

\renewcommand{\P}[1]{\mathbb{P}\bs{{#1}}}
\newcommand{\Pp}[2]{\underset{#1}{\mathbb{P}}\bs{{#2}}}

\newcommand{\E}{\mathbb{E}}
\DeclareMathOperator*{\argmin}{arg\,min}

\newcommand{\trans}{\intercal}

\newcommand{\innerprod}[2]{\left\langle{#1},{#2}\right\rangle}

\newcommand{\cost}{c}

\usepackage{xcolor}

\newcommand{\fcost}{x}

%% file: smoothness_proof.tex
We start by taking a first partial derivative of the potential function for the fractional cost $x_{\bar{i}e}$, that is
\begin{equation*}
    \frac{\partial }{\partial x_{\bar{i}e}} \bs{\Phi(x)} = \sum_{\mathcal{S}\subset[n], \bar{i}\in\mathcal{S}} \prod_{j\in\mathcal{S}, j \neq \bar{i}} x_{je} \prod_{j\notin\mathcal{S}} (1 - x_{je}) \sum^{\abs{\mathcal{S}}}_{i=0} c_e(i) - \sum_{\mathcal{S}\subset[n], \bar{i}\notin\mathcal{S}} \prod_{j\in\mathcal{S}} x_{je} \prod_{j\notin\mathcal{S}, j \neq \bar{i}} (1 - x_{je}) \sum^{\abs{\mathcal{S}}}_{i=0} c_e(i)
\end{equation*}

Taking a second a partial derivative with respect to $x_{\bar{j}\bar{e}}$ we obtain
\begin{equation*}
    \frac{\partial^2 }{\partial x_{\bar{j}\bar{e}} \partial x_{\bar{i}e}} \bs{\Phi(x)} = \begin{cases} 0 \quad \text{if} \quad \bar{e} \neq e \\
    0 \quad \text{if} \quad \bar{i} = \bar{j} \\
    \sum_{\mathcal{S}\subset[N], \bar{i}\in\mathcal{S}, \bar{j}\in\mathcal{S}} \prod_{j\in\mathcal{S}, j \neq \bar{i}, j \neq \bar{j}} x_{je} \prod_{j\notin\mathcal{S}} (1 - x_{je}) \sum^{\abs{\mathcal{S}}}_{i=0} c_e(i) \\ - \sum_{\mathcal{S}\subset[N], \bar{i}\notin\mathcal{S}, \bar{j}\in\mathcal{S}} \prod_{j\in\mathcal{S}, j \neq \bar{j}} x_{je} \prod_{j\notin\mathcal{S},  j \neq \bar{i}} (1 - x_{je}) \sum^{\abs{\mathcal{S}}}_{i=0} c_e(i) \\ - \sum_{\mathcal{S}\subset[N], \bar{i}\in\mathcal{S}, \bar{j}\notin\mathcal{S}} \prod_{j\in\mathcal{S}, j \neq \bar{i}} x_{je} \prod_{j\notin\mathcal{S},  j \neq \bar{j}} (1 - x_{je}) \sum^{\abs{\mathcal{S}}}_{i=0} c_e(i) \\ + \sum_{\mathcal{S}\subset[N], \bar{i}\notin\mathcal{S}, \bar{j}\notin\mathcal{S}} \prod_{j\in\mathcal{S}} x_{je} \prod_{j\notin\mathcal{S}, j \neq \bar{i}, j \neq \bar{j}} (1 - x_{je}) \sum^{\abs{\mathcal{S}}}_{i=0} c_e(i) \quad \text{otherwise}
    \end{cases}
\end{equation*}

Notice that the factor $\prod_{j\in\mathcal{S}, j \neq \bar{j}} x_{je} \prod_{j\notin\mathcal{S},  j \neq \bar{i}} (1 - x_{je})$ is a multivariate probability distribution over the subsets of players in which the ones with index $\bar{j}$ and $\bar{i}$ select the edge $e$ therefore the sum $\sum_{\mathcal{S}\subset[n], \bar{i}\in\mathcal{S}, \bar{j}\in\mathcal{S}} \prod_{j\in\mathcal{S}, j \neq \bar{i}, j \neq \bar{j}} x_{je} \prod_{j\notin\mathcal{S}} (1 - x_{je}) = 1$. Similar observations hold for the other three terms in the nonzero partial derivatives. Then it follows that $-2 n c_{\max} \leq \frac{\partial^2 }{\partial x_{\bar{j}\bar{e}} \partial x_{\bar{i}e}} \bs{\Phi(x)} \leq 2 n c_{\max} $ where $c_{\max}$ is a uniform bound on $c_e(\cdot)$.

Furthermore, we observe that there are at most $mn(n - 1)$ nonzero elements of the Hessian and those satisfy $\abs{\frac{\partial^2 }{\partial x_{\bar{j}\bar{e}} \partial x_{\bar{i}e}} \bs{\Phi(x)}}^2 \leq 4 n^2 c_{\max}^2$. Therefore, we can bound the Hessian Frobenius norm as
\begin{align*}
    \norm{\nabla^2 \Phi(x)}_F &= \sqrt{\sum_{e \in E}\sum_{\bar{e} \in E}\sum^N_{\bar{j}=1} \sum^N_{\bar{i}=1} \abs{\frac{\partial^2 }{\partial x_{\bar{j}\bar{e}} \partial x_{\bar{i}e}} \bs{\Phi(x)}}^2} \\ & \leq \sqrt{mn(n - 1) 4 n^2 c_{\max}^2 } \\ & \leq 2 n^2  c_{\max} \sqrt{m}
\end{align*}
Finally, since we have that for any matrix the Frobenius norm upper bounds the spectral norm, we can conclude that the maximum eigenvalue of $\nabla^2 \Phi(x)$ is at most $2 n^2  c_{\max} \sqrt{m}$ and therefore the potential function is $2 n^2  c_{\max} \sqrt{m}$-smooth.